\documentclass[11pt]{article}
\usepackage[utf8]{inputenc}
\usepackage[T1]{fontenc}
\usepackage[nocompress]{cite}

\newcommand{\clustering}{\mathcal{C}}

\usepackage{amsmath}
\usepackage{amsthm}
\usepackage{amsfonts}
\usepackage{fullpage}
\usepackage{thm-restate}
\usepackage{enumitem}
\usepackage{subcaption}

\usepackage[linktocpage]{hyperref}
\usepackage[dvipsnames]{xcolor}
\definecolor{ForestGreen}{rgb}{0.0333,0.4451,0.0333}
\definecolor{DarkRed}{rgb}{0.65,0,0}
\definecolor{Red}{rgb}{1,0,0}
\hypersetup{
    unicode=false,          
    colorlinks=true,        
    linkcolor=BrickRed,          
    citecolor=OliveGreen,        
    filecolor=magenta,      
    urlcolor=cyan           
}

\usepackage[linesnumbered,boxed, vlined]{algorithm2e}
	\DontPrintSemicolon
  \SetKwInOut{Input}{Input}
  \SetKwInOut{Output}{Output}
  \SetKwProg{function}{Function}{:}{}

\usepackage{algorithmic}

\usepackage{graphicx}
\usepackage[capitalise]{cleveref}
\usepackage{xspace}
\usepackage{xcolor}

\newcommand{\am}{\textsf{Adaptive Minimum}\xspace}

\newcommand{\hac}{\textsf{Average Linkage HAC}\xspace}
\newcommand{\lfm}{\textsf{LFM Matching}\xspace}

\DeclareMathOperator*{\argmax}{arg\,max}
\DeclareMathOperator*{\argmin}{arg\,min}
\newcommand{\CC}{\textsc{CC}\xspace}
\newcommand{\NC}{\textsc{NC}\xspace}
\newcommand{\PTime}{\textsc{P}\xspace}
\newcommand{\poly}{\ensuremath{\mathsf{poly}}}
\newcommand{\polylog}{{\ensuremath\mathsf{polylog}}}

\newtheorem{problem}{Problem}
\newtheorem{theorem}{Theorem}
\newtheorem{conjecture}{Conjecture}
\newtheorem{claim}{Claim}

\newtheorem{lemma}{Lemma}
\newtheorem{definition}{Definition}

\newcommand{\bestedge}[1]{\ensuremath{\textsc{BestEdge}(#1)}}

\newcounter{mnote}[section]

\title{
It's Hard to HAC with Average Linkage!
}
\author{%
MohammadHossein Bateni\\%
Google Research\\%
New York, USA%
\and%
Laxman Dhulipala\\%
University of Maryland\\%
College Park, USA%
\and%
Kishen N Gowda\\%
University of Maryland\\
College Park, USA%
\and%
D Ellis Hershkowitz\\%
Brown University\\
Providence, USA%
\and%
Rajesh Jayaram\\%
Google Research\\
New York, USA%
\and%
Jakub Łącki\\%
Google Research\\
New York, USA%
}
\date{}

\begin{document}
\maketitle

\begin{abstract}
Average linkage Hierarchical Agglomerative Clustering (HAC) is an extensively studied and applied method for hierarchical clustering. Recent applications to massive datasets have driven significant interest in near-linear-time and efficient parallel algorithms for average linkage HAC.

We provide hardness results that rule out such algorithms. On the sequential side, we establish a runtime lower bound of $n^{3/2-\epsilon}$ on $n$ node graphs for sequential combinatorial algorithms under standard fine-grained complexity assumptions. This essentially matches the best-known running time for average linkage HAC. On the parallel side, we prove that average linkage HAC likely cannot be parallelized even on simple graphs by showing that it is CC-hard on trees of diameter $4$. On the possibility side, we demonstrate that average linkage HAC can be efficiently parallelized (i.e., it is in NC) on paths and can be solved in near-linear time when the height of the output cluster hierarchy is small.
\end{abstract}
\thispagestyle{empty}
\newpage
\setcounter{page}{1}

\section{Introduction}

Hierarchical clustering is a fundamental method for data analysis which organizes data points into a hierarchical structure so that similar points appear closer in the hierarchy. Unlike other common clustering methods, such as $k$-means, hierarchical clustering does not require the the number of clusters to be fixed ahead of time. This allows it to capture structures that are inherently hierarchical---such as phylogenies \cite{eisen1998cluster} and brain structure \cite{boly2012hierarchical}. 
One of the most widely used and studied methods for hierarchical clustering is Hierarchical Agglomerative Clustering (HAC) \cite{king1967step,lance1967general,sneath1973principles}. HAC produces a hierarchy by first placing each point in its own cluster and then iteratively merging the two \emph{most similar} clusters until all points are aggregated into a single cluster. The similarity of two clusters is given by a \emph{linkage function}. HAC is included in many popular scientific computing libraries such as scikit-learn \cite{pedregosa2011scikit}, SciPy \cite{virtanen2020scipy}, ALGLIB \cite{shearer1985alglib}, Julia, R, MATLAB, Mathematica  and many more \cite{murtagh2012algorithms,murtagh2017algorithms}. This cluster hierarchy is often equivalently understood as a binary tree---a.k.a. dendrogram---whose internal nodes correspond to cluster merges.

The proliferation of massive datasets with billions of points has driven the need for more efficient HAC algorithms that can overcome the inherent $\Theta(n^2)$ complexity required to read all pairwise distances~\cite{dhulipala2023terahac, dhulipala2022hierarchical, monath2021scalable}. 
Finer-grained running time bounds for HAC were recently obtained by assuming that only $m = o(n^2)$ pairs of points have nonzero similarity, and analyzing the running time as a function of both $n$ and $m$.
This is a natural assumption in practice, as in large datasets of billions of datapoints, typically a small fraction of pairs exhibit nonnegligible similarity.
In this case, the input to HAC is an edge-weighted graph, where each vertex represents an input point and each edge weight specifies the similarity between its endpoints.
This approach is convenient for large-scale applications since (1) very large clustering instances can be compactly represented as sparse weighted graphs and (2) the running time of HAC can be decoupled from the running time of nearest-neighbor search.
%


A particularly common linkage function for HAC is \emph{average linkage}, which both optimizes reasonable global objectives \cite{moseley2023approximation} and exhibits good empirical performance \cite{irbook, affinity, mllner2011modern, zhao2002evaluation, hua2017mgupgma, kobren2017hierarchical, moseley-wang, hac-reward, monath2021scalable}.
Here, the similarity of two clusters is the average edge weight between them (non-present edges are treated as having weight $0$). In other words, average linkage HAC repeatedly merges the two clusters with the highest average edge weight between them (see \Cref{fig:HAC} for an example).

\begin{figure}[h]
    \centering
    \begin{subfigure}[b]{0.19\textwidth}
        \centering
        \includegraphics[width=\textwidth,trim=0mm 0mm 350mm 0mm, clip]{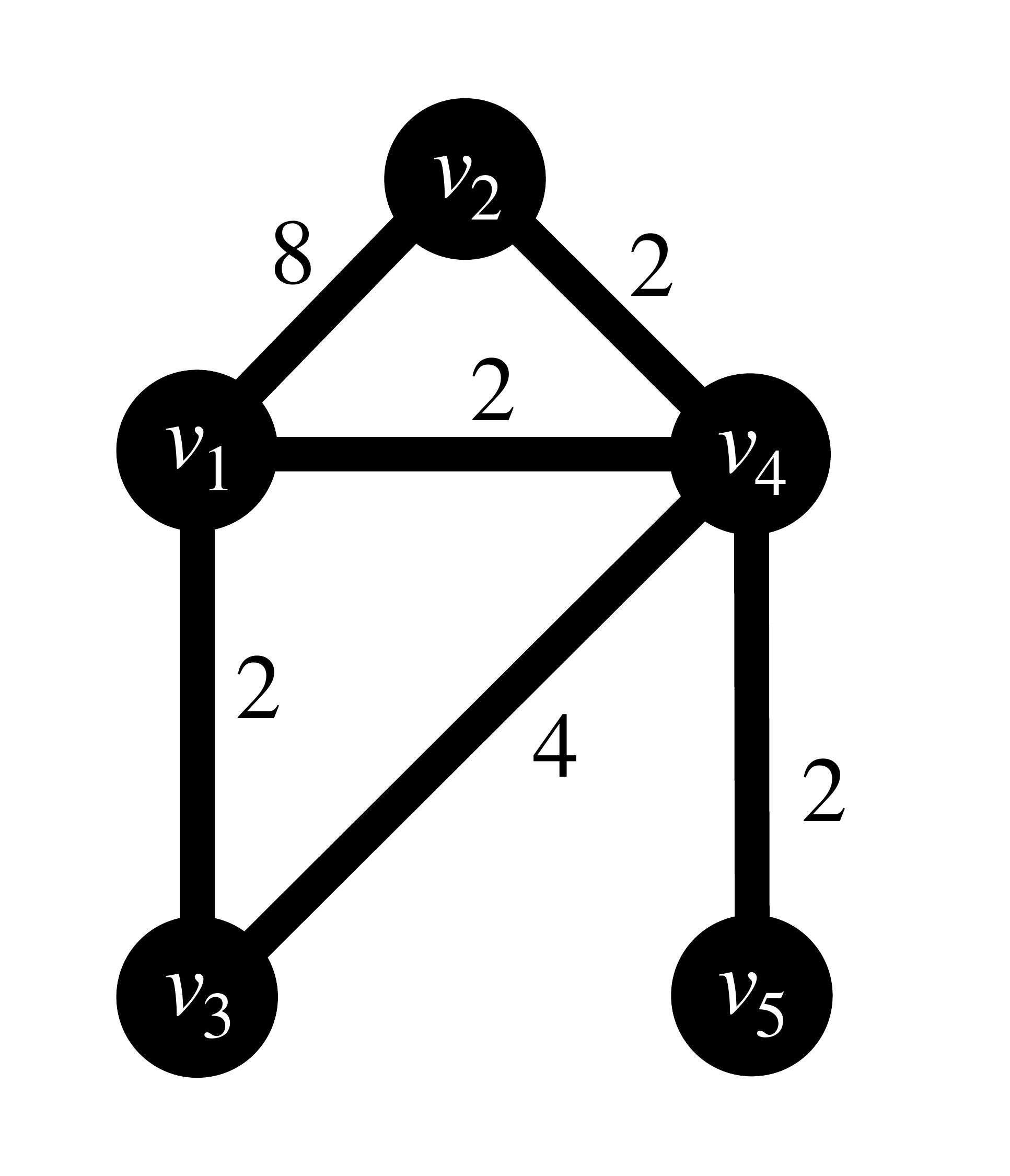}
        \caption{Input $G$.}\label{sfig:hac1}
    \end{subfigure} \hspace{3.5em} 
    \begin{subfigure}[b]{0.19\textwidth}
        \centering
        \includegraphics[width=\textwidth,trim=0mm 0mm 350mm 0mm, clip]{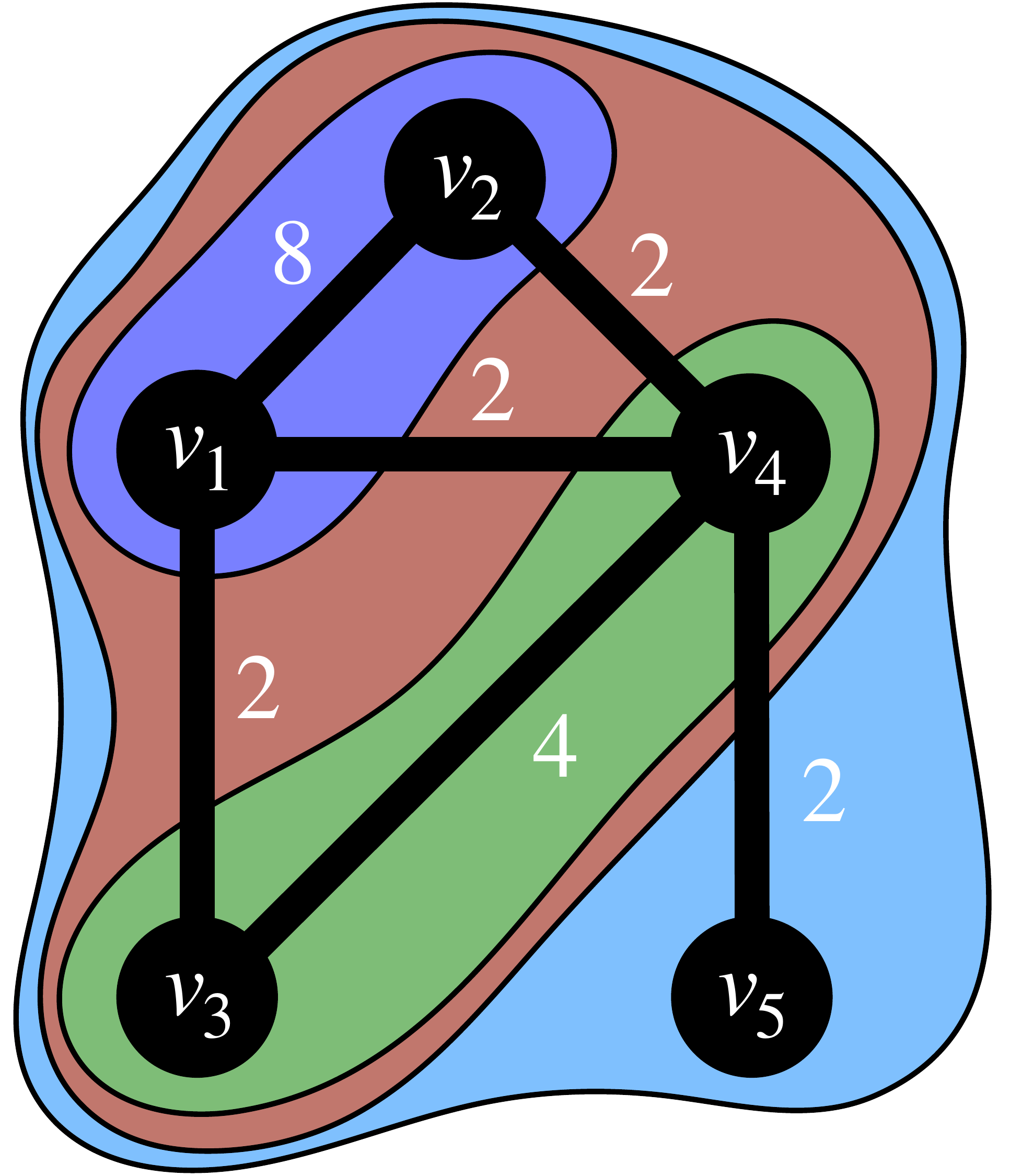}
        \caption{HAC Output.}\label{sfig:hac2}
    \end{subfigure}  \hspace{3.5em}
    \begin{subfigure}[b]{0.19\textwidth}
        \centering
        \includegraphics[width=\textwidth,trim=0mm 0mm 320mm 0mm, clip]{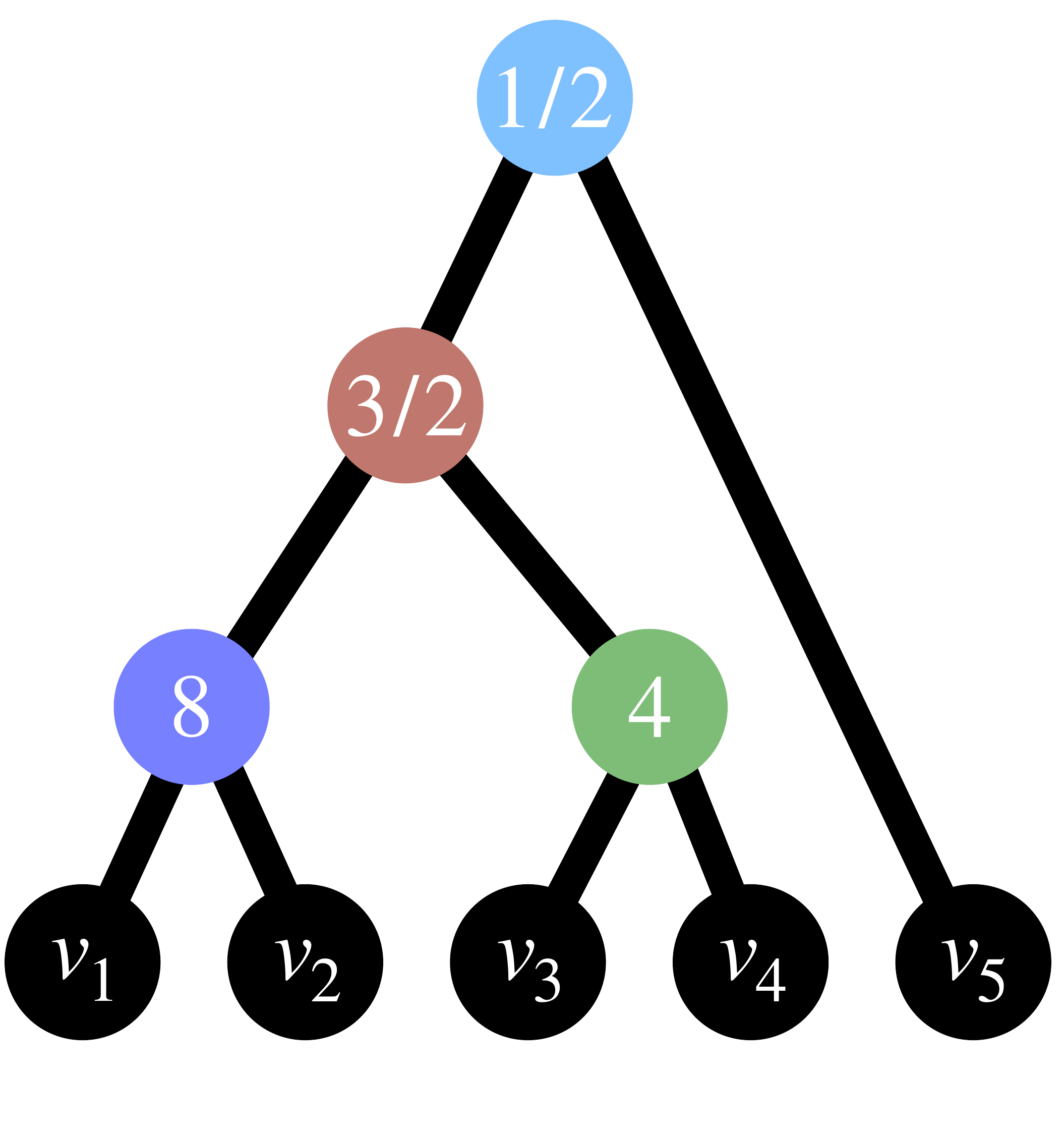}
        \caption{Dendrogram.}\label{sfig:hac3}
    \end{subfigure}   
    \vspace{-0.5em}
    \caption{An example of average linkage HAC run on an input graph $G$. Edges labeled with weights. \ref{sfig:hac1} gives $G$. \ref{sfig:hac2} gives the cluster hierarchy output by HAC. \ref{sfig:hac3} gives the corresponding dendrogram with internal nodes labeled with the weight of their corresponding merge.} \label{fig:HAC}
\end{figure}


A natural algorithmic question then is how quickly can we solve average linkage HAC on $n$ node and $m$ edge graphs? 
Recent work has provided a partial answer to this basic question in sequential and parallel models of computation.
In particular, \cite{dhulipala2021hierarchical} showed that average linkage HAC can be solved in $\tilde{O}(n\sqrt{m})$ time, thus providing a sub-quadratic time algorithm for sufficiently sparse graphs.
A follow-up paper studied average linkage HAC in the parallel setting and showed that the problem is \PTime-complete and so likely does not admit \NC algorithms~\cite{dhulipala2022hierarchical}. 
However, the \PTime-completeness result of \cite{dhulipala2022hierarchical} holds for worst case graphs whereas typical applications of HAC are on highly structured graphs---namely those which are meant to capture relevant properties of an underlying metric---and so there is still hope for parallelizing average linkage HAC on more structured instances. 

In fact, such structured instances of average linkage HAC are known to admit much faster algorithms in the sequential setting: the sequential algorithm of \cite{dhulipala2021hierarchical} implies that if the input graph is planar (or, more generally, minor-free) average linkage HAC can be solved in time $\tilde{O}(m)$. 
More generally, if each graph obtained by contracting all clusters at each step of average linkage HAC has $O(1)$ arboricity\footnote{A graph has arboricity at most $\alpha$ if all of its edges can be covered by at most $\alpha$ trees.}, then it is possible to solve average linkage HAC in time $\tilde{O}(m)$; it follows that average linkage  HAC can be solved in sequential time $\tilde{O}(m)$ on trees or planar graphs.
In light of these improved sequential results for highly structured graphs, it becomes natural to hope for efficient parallel algorithms on structured graphs such as low arboricity graphs or, even, just trees.





\subsection{Our Contributions}
In this work, we continue the line of work which studied the computational complexity of different variants of HAC~\cite{greenlaw2008parallel, tsenghac, abboudhac, dhulipala2021hierarchical, dhulipala2022hierarchical} and perform a careful investigation into the complexity of average linkage HAC.
In particular, we study HAC on $n$ node and $m$ edge graphs and investigate whether
near-linear time algorithms, or more efficient parallel algorithms are possible, namely:
\begin{enumerate}
    \item \textbf{Near-Linear Time Algorithms:} Can we improve over the best known $\tilde{O}(n\sqrt{m})$ upper bound for average linkage HAC and obtain near-linear time sequential algorithms?
    
    \item \textbf{\NC Algorithms:} are there $\polylog(n)$ depth parallel algorithms for average linkage HAC with $\poly(n)$ work for highly structured instances, e.g., trees, or minor-closed graphs?
\end{enumerate}
We give both new lower bounds which (conditionally) rule out near-linear time and \NC algorithms, and provide conditions under which these impossibility results can be bypassed.

First, we demonstrate that near-linear time algorithms are impossible under standard fine-grained complexity assumptions. 
\begin{restatable}{theorem}{mainLowerSeq}\label{thm:worklb}
    If average linkage HAC can be solved by a combinatorial algorithm in $O(n^{3/2-\epsilon})$ time for any $\epsilon > 0$, then the Combinatorial Boolean Matrix Multiplication (Combinatorial BMM) Conjecture is false.
\end{restatable}

Our reduction also implies a second (weaker) conditional lower bound that also holds for non-combinatorial algorithms (e.g., algebraic algorithms) based on the running time of matrix multiplication. 
In particular, for two $n \times n$ binary matrices, it is well known that matrix multiplication can be solved in time $O(n^{\omega})$ where $2 \leq \omega < 2.3716$~\cite{williams2024new}.
In this setting, we obtain the following result:

\begin{restatable}{theorem}{workmatmul}\label{cor:workmatmul}
If average linkage HAC can be solved by an algorithm in $O(n^{\omega/2-\epsilon})$ time for some $\epsilon > 0$, then boolean matrix multiplication can be solved in $O(n^{\omega - \epsilon'})$ time for some $\epsilon' > 0$.
\end{restatable}

\noindent 
Notably, Theorem~\ref{thm:worklb} shows that the prior running time of $\tilde{O}(n \sqrt{m})$ of \cite{dhulipala2021hierarchical} is {\em optimal} up to logarithmic factors under standard fine-grained complexity assumptions, at least for graphs consisting of $O(n)$ many edges.
We obtain this conditional lower bound by showing that a carefully constructed instance of HAC can be used to solve the triangle detection problem, which is sub-cubically equivalent to Boolean Matrix Multiplication~\cite{williams2010subcubic}.
We obtain a bound of (essentially) $\Omega(n^{3/2})$ since our reduction incurs a quadratic time and space blowup when transforming an input triangle detection instance to an instance of average linkage HAC.

We next turn to the parallel setting. 
Here, we show that HAC---\emph{even on trees}---is unlikely to admit efficient parallel algorithms. 
More formally, we show that average linkage HAC on low diameter trees is as hard as any problem in the complexity class Comparator Circuit (\CC) \cite{subramanian1989new, cook2014complexity}. 
It is believed that \CC is incomparable with \NC and that \CC-hardness is evidence that a problem is not parallelizable \cite{mayr1992complexity, cook2014complexity}.
\begin{restatable}{theorem}{mainLowerPar}\label{thm:mainLowerPar}
Average linkage HAC is \CC-hard, even on trees of diameter $4$.
\end{restatable}
\noindent We note that it is known that $\CC \subseteq \PTime$ and so the $\PTime$-hardness of \cite{dhulipala2022hierarchical} already suggests the impossibility of efficient parallel algorithms on \emph{general graphs}. 
However, our result suggests the impossibility of efficient parallel algorithms \emph{even on very simple graphs} (trees of diameter 4).
We obtain this result by reducing from the lexicographically first maximal matching (\lfm) problem and an intermediate problem which we call \am, which captures some of what makes HAC intrinsically difficult to parallelize.

On the positive side, we demonstrate that average linkage HAC on path graphs is in \NC, under the mild assumption that the aspect ratio is polynomial.
While the class of path graphs is restrictive, even on paths average linkage is highly non-trivial and naively running HAC requires resolving chains of $\Omega(n)$ sequential dependencies. For example, consider a path of vertices $(v_1, v_2, \ldots, v_n)$ where the edge $\{v_{i}, v_{i+1}\}$ has weight $1+i \cdot \epsilon$ for some small $\epsilon > 0$ and initially each vertex is in its own cluster. Initially, $v_n$'s most similar neighbor is $v_{n-1}$ and so $v_n$ would like to merge with $v_{n-1}$ but $v_{n-1}$'s most similar neighbor is $v_{n-2}$ and so on. Thus, whether or not $v_n$ gets to merge with $v_{n-1}$ depends on the merge behavior of $\Theta(n)$ other clusters and so it is not at all clear that \NC algorithms should be possible for this setting. Nonetheless, we show the following.
\begin{restatable}{theorem}{mainUpperPar}\label{thm:mainUpperPar}
    Average linkage HAC on paths is in \NC. In particular, there is an algorithm for average linkage HAC that runs in $O(\log^2 n \log\log n)$ depth with $O(n\log n\log\log n)$ work.
\end{restatable}
\noindent  The above algorithm leverages the fact that in average linkage HAC the maximum edge similarity monotonically decreases. In particular, it works in $O(\log n)$ phases where each phase consists of merges of equal similarity up to constants. The goal then becomes to efficiently perform merges until every edge is no longer within a constant of the starting maximum similarity of the phase. The starting point of the algorithm is to observe that $\Omega(n)$ sequential dependencies of clusters of equal size can be resolved efficiently in parallel in a phase by noting that in this phase only the odd-indexed edges merge in the chain. Thus, each edge can decide if it is odd-indexed in parallel by, e.g., using $\mathsf{prefix\text{-}sum}$, which is well known to be solvable in linear work in \NC.

For chains with clusters of general weights, we decompose dependency chains into short ($(O(\log n)$-length) subchains where resolving dependencies within the subchain must be done sequentially but in the current phase each subchain's merge behavior only depends on whether or not its closest neighboring subchains merges into it or not. Thus, each subchain can compute its merge behavior for these two cases and then, similar to the equal weights setting, we propagate merge behavior across subchains efficiently in parallel.

To complement our sequential lower bound with a positive result, we demonstrate that it is possible to achieve near-linear running time, provided the dendrogram has low height. Thus, if the output dendrogram is a relatively balanced tree, then near-linear time algorithms are possible.
\begin{restatable}{theorem}{mainUpperSec}\label{thm:mainUpperSeq}
    There is an implementation of the nearest-neighbor chain algorithm for average linkage HAC that runs in $O(m \cdot h \log n)$ time where $h$ is the height of the output dendrogram.
\end{restatable}
\noindent The above result is in fact obtained by a relatively simple (but to the best of our knowledge, new) analysis of {\em existing} classic HAC algorithms.
In particular, we show that the nearest-neighbor chain~\cite{benzecri, juan} and heap-based algorithms~\cite{irbook} for HAC, which were developed over 40 years ago achieve this bound. 


\section{Preliminaries}\label{sec:prelims}
The input to the HAC algorithm is an undirected weighted graph $G=(V, E, w)$, where $w : V \times V \rightarrow \mathbb{R}_+ \cup \{0\}$ is a function assigning nonnegative weights to the edges.
For convenience we assume $w(x, y) = 0$ when $xy \not\in E$.
The vanilla version of average linkage HAC is given as \cref{alg:staticgraph}.
It starts by putting each vertex in a cluster of size $1$ and then repeats the following step.
While there is a pair of clusters of positive similarity, find two most similar clusters and merge them together, that is, replace them by their union.
The similarity between two clusters is the total edge weight between them divided by the product of the cluster sizes.
We refer to this version as the \emph{static graph} version, since the graph is not changed throughout the run of the algorithm.

Throughout the paper we usually work with a different (equivalent) way of presenting the same algorithm which is given as \cref{alg:contractions} (e.g., \Cref{alg:nnchain,alg:heapbased}).
In this version we maintain a graph $G$ whose vertices are clusters.
The \emph{size} of the vertex is the size of the cluster it represents.
The \emph{normalized} weight of an edge $xy$ in $G$ is $w(x,y)$ divided by the product of the sizes of $x$ and $y$.

Whenever two clusters merge, their corresponding vertices are merged into one, i.e., the edge between them is contracted and the size of the new vertex is the sum of the sizes of the vertices that merged.
In the following we sometimes say that a vertex $x$ merges into vertex $y$.
In this case we simply assume that the name of the resulting vertex is $y$ and the size of $y$ is increased by the size of $x$. See \Cref{fig:HACMerge}.

\begin{figure}[h]
    \centering
    \vspace{-0.5em}
    \begin{subfigure}[b]{0.19\textwidth}
        \centering
        \includegraphics[width=\textwidth,trim=0mm 0mm 350mm 0mm, clip]{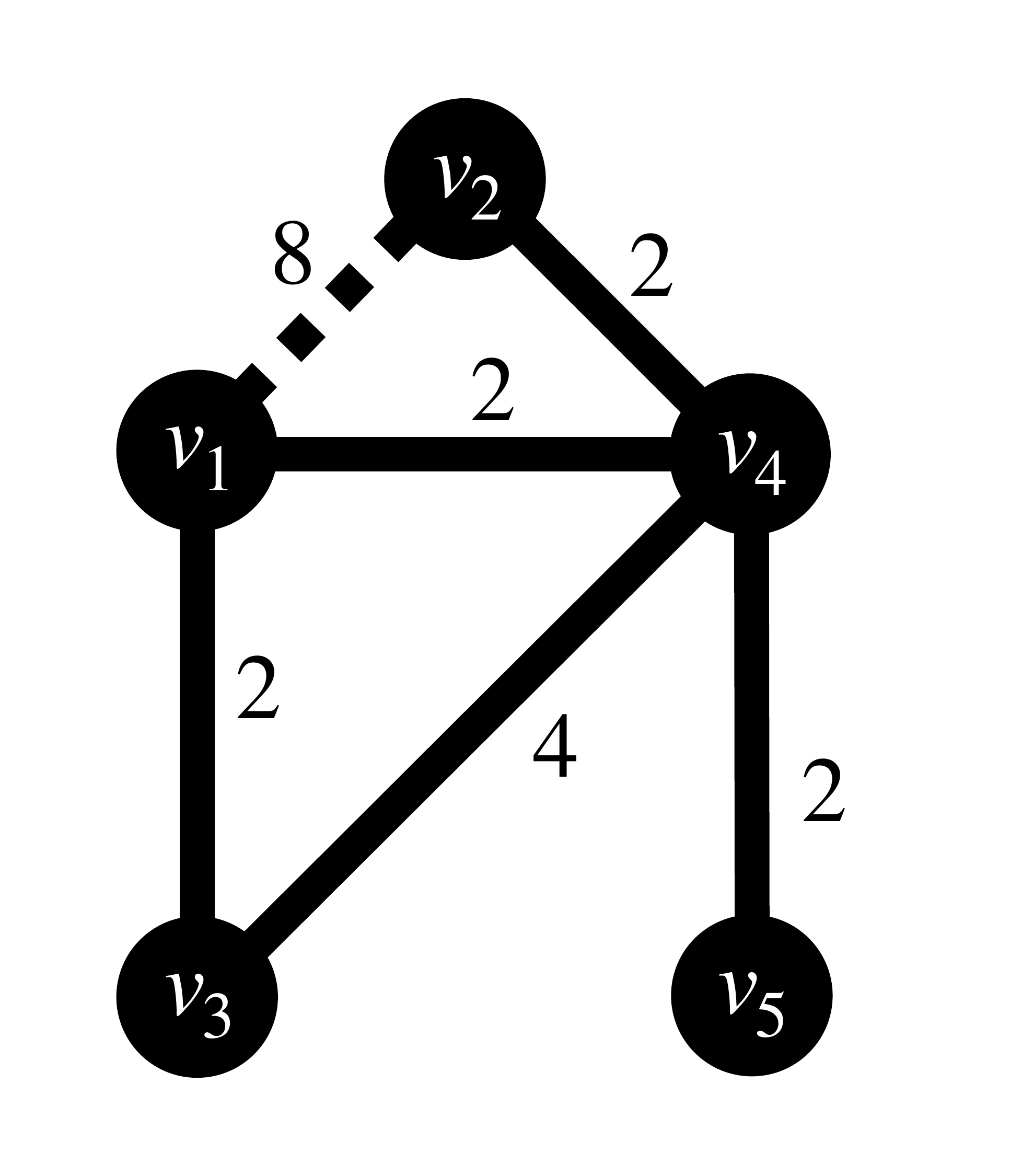}
        \caption{Input $G$.}\label{sfig:hacMerge1}
    \end{subfigure} \hfill
    \begin{subfigure}[b]{0.19\textwidth}
        \centering
        \includegraphics[width=\textwidth,trim=0mm 0mm 350mm 0mm, clip]{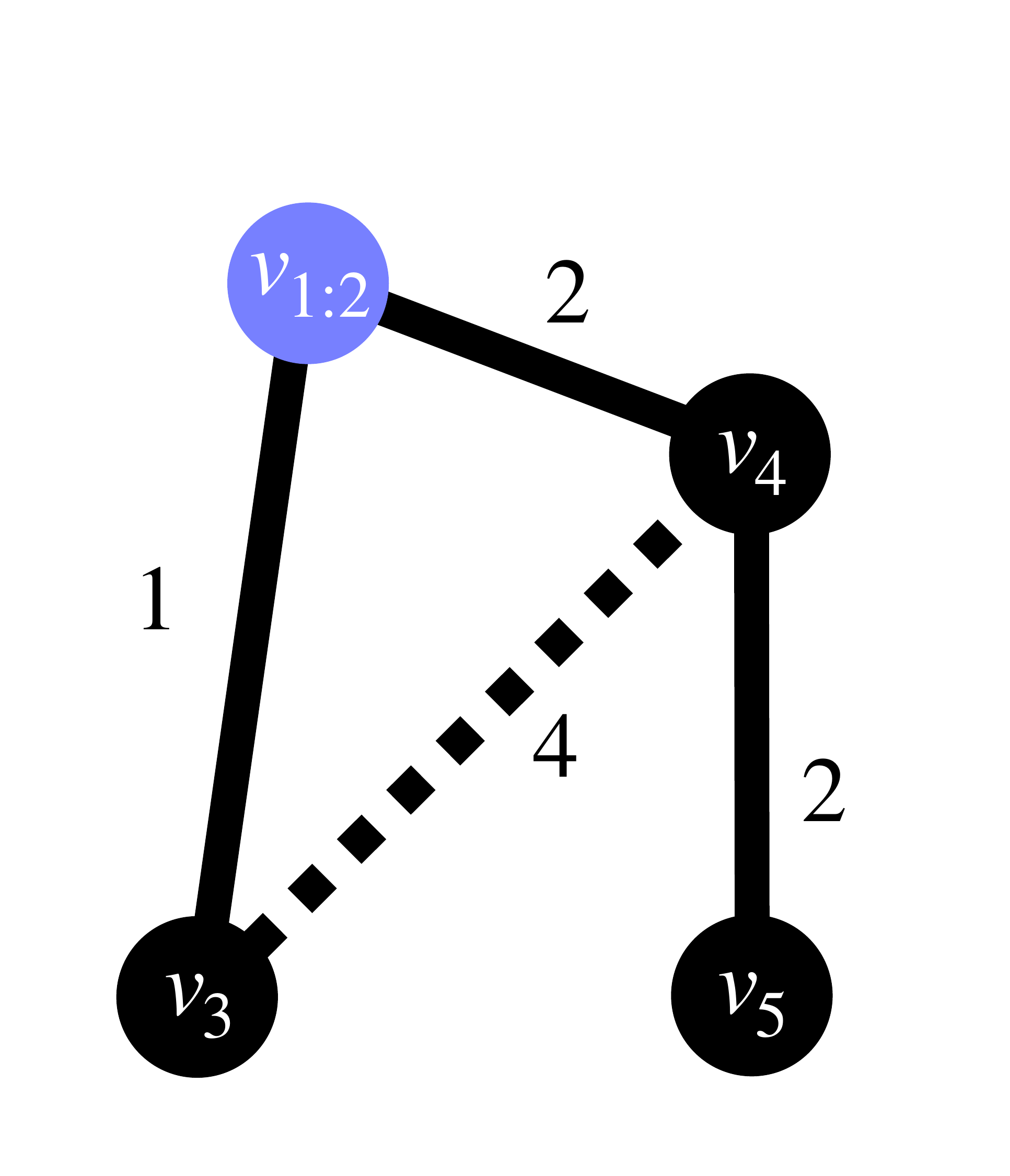}
        \caption{1 merge.}\label{sfig:hacMerge2}
    \end{subfigure}  \hfill
    \begin{subfigure}[b]{0.19\textwidth}
        \centering
        \includegraphics[width=\textwidth,trim=0mm 0mm 320mm 0mm, clip]{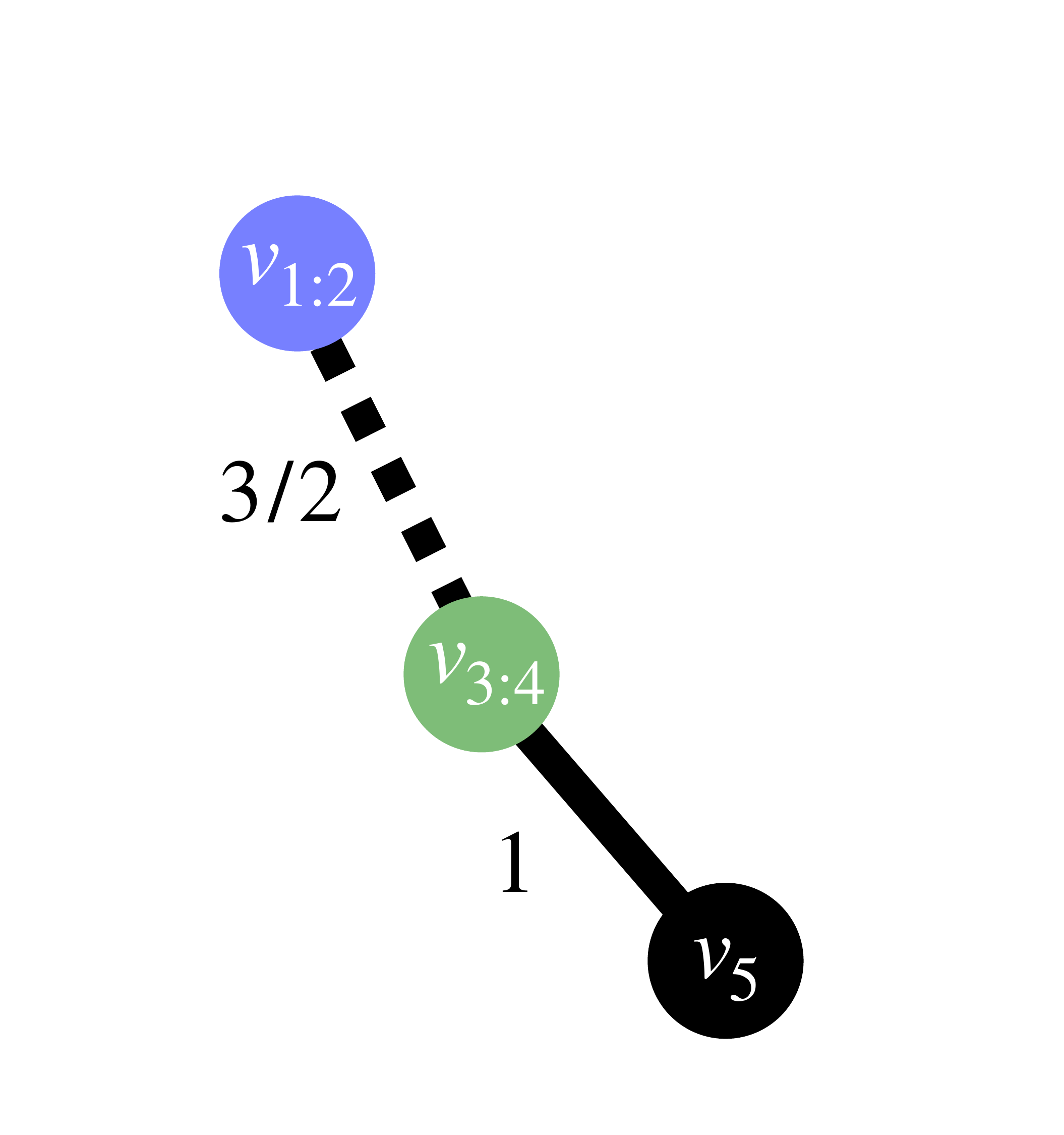}
        \caption{2 merges.}\label{sfig:hacMerge3}
    \end{subfigure}   \hfill
        \begin{subfigure}[b]{0.19\textwidth}
        \centering
        \includegraphics[width=\textwidth,trim=0mm 0mm 320mm 0mm, clip]{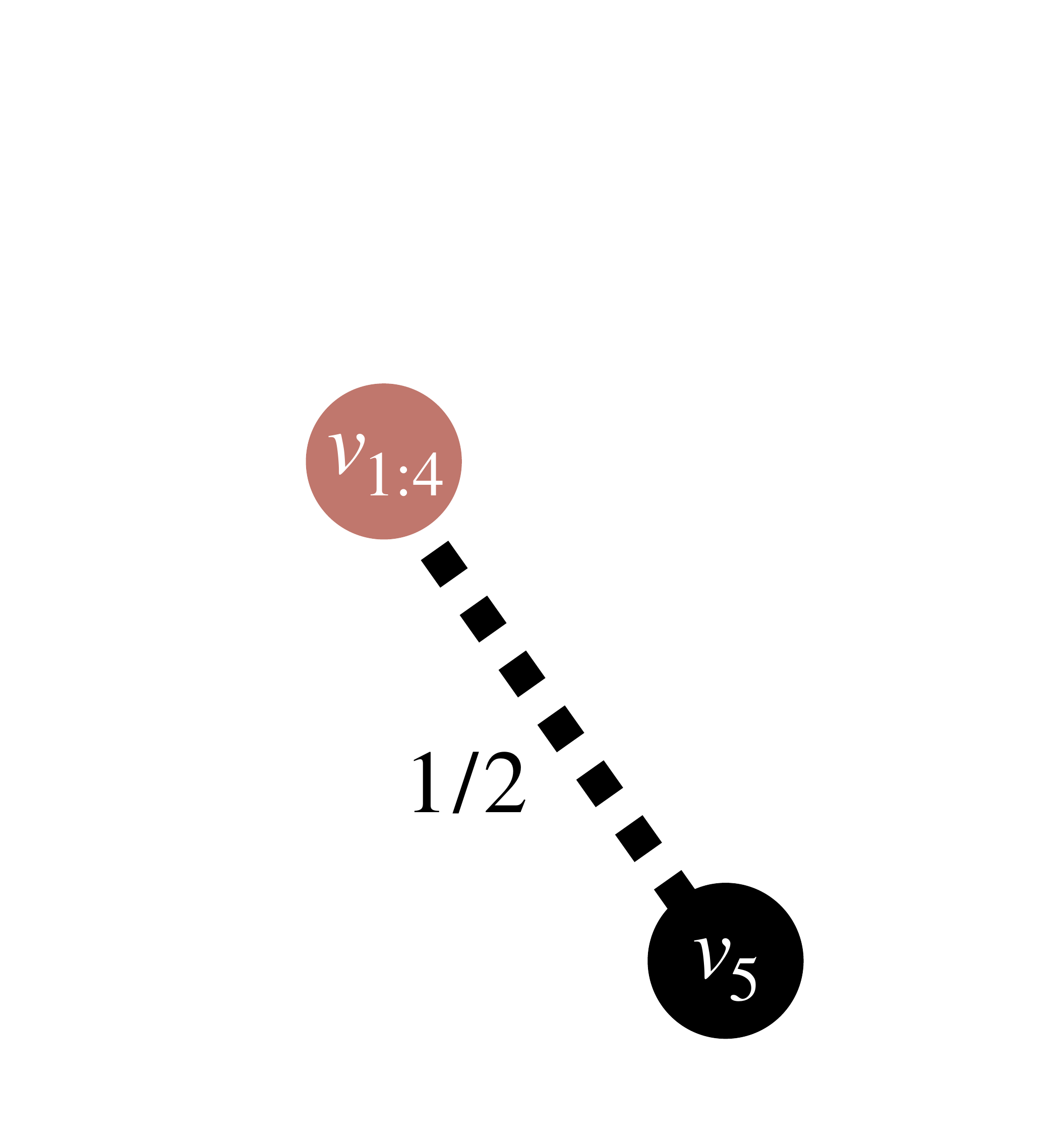}
        \caption{3 merges.}\label{sfig:hacMerge4}
    \end{subfigure}   \hfill
        \begin{subfigure}[b]{0.19\textwidth}
        \centering
        \includegraphics[width=\textwidth,trim=0mm 0mm 320mm 0mm, clip]{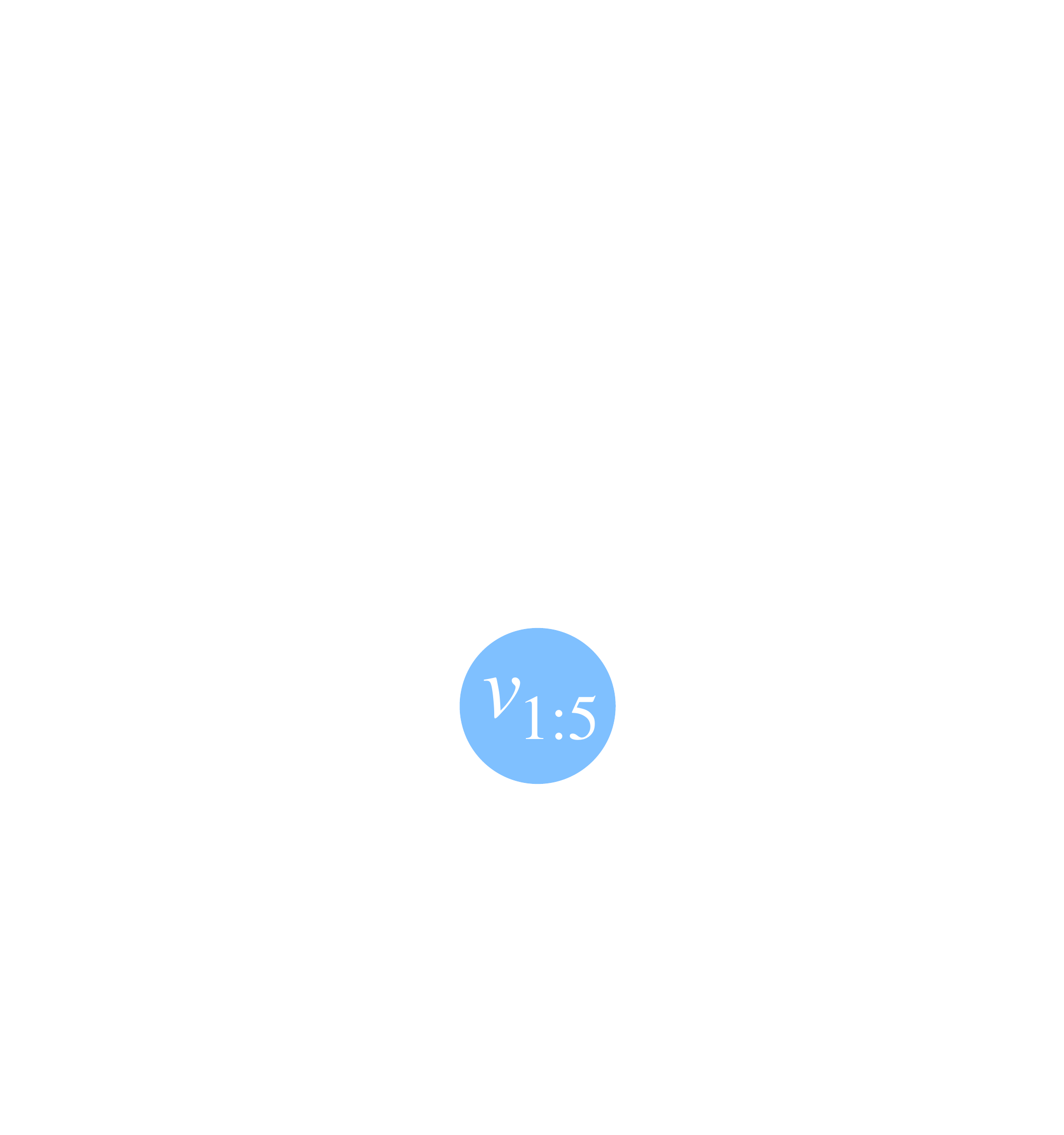}
        \caption{4 merges.}\label{sfig:hacMerge5}
    \end{subfigure}   
    \vspace{-0.5em}
    \caption{An example of average linkage HAC run on an input graph $G$ where we imagine we contract merged clusters. Intermediate vertices labeled with the vertices of $G$ their corresponding cluster contains. Edges labeled with their weight and next merged edge is dashed. } \label{fig:HACMerge}
\end{figure}

The output of HAC is a \emph{dendrogram}---a rooted binary tree representing the cluster merges performed by the algorithm.
Every node of the dendrogram is a cluster built by the algorithm.
There are exactly $|V|$ leaves corresponding to the single-element clusters that are formed in the beginning of the algorithm.
Whenever two clusters $C_1$ and $C_2$ are merged, we add to the dendrogram a new node $C_1 \cup C_2$ whose children are $C_1$ and $C_2$. See \Cref{sfig:hac3} for the dendrogram of \Cref{fig:HACMerge}.

\newcommand{\forkins}{\texttt{fork}}
\newcommand{\insend}{\texttt{end}}
\newcommand{\allocateins}{\texttt{allocate}}
\newcommand{\freeins}{\texttt{free}}
\newcommand{\mpram}{MT-RAM}
\newcommand{\thread}{thread}
\newcommand{\bfmodel}{binary-forking model}

We use the classic multithreaded model~\cite{BL98,ABP01,blelloch2020optimal} (formally, the MP-RAM~\cite{blelloch2020optimal}) to analyze the parallel algorithms.
We assume a set of \thread{}s that share the memory.
Each \thread{} acts like a sequential RAM plus a \forkins{} instruction  that forks two new child threads.
When a thread performs a fork, the two child threads can both start by running their next instructions, and the original thread is suspended until both children terminate.
A computation starts with a single root thread and finishes when that root thread finishes.
A parallel for-loop can be viewed as executing \forkins{}s for a logarithmic number of levels.
A computation can thus be viewed as a DAG (directed acyclic graph).
We say the {\em work} is the total number of operations in this DAG and {\em span (depth)} is equal to the longest path in the DAG.
We note that computations in this model can be cross-simulated in standard variants of the PRAM model in the same work (asymptotically), and losing at most a single logarithmic factor in the depth~\cite{blelloch2020optimal}.

\begin{algorithm}[t]
\caption{Average linkage HAC---static graph version.}\label{alg:staticgraph}
\KwIn{$G = (V, E, w)$}
\SetKwFunction{FSim}{Similarity}
\SetKwFunction{FHac}{HAC}
\SetKwProg{Fn}{Function}{:}{}

\Fn{\FSim{$C_1, C_2, w$}}{
\Return $\sum_{x \in C_1, y \in C_2} w(x,y) / (|C_1| \cdot |C_2|)$
}

\Fn{\FHac{$G$}}{
$\clustering \gets$ clustering where each vertex of $G$ is in a separate cluster\\
\While{$\exists_{C_1, C_2 \in \clustering}$ s.t. $C_1 \neq C_2$ and $\FSim(C_1, C_2, w) > 0$}{
$(C_1, C_2) = \argmax_{(C_1, C_2) \in \clustering \times \clustering} \FSim(C_1, C_2, w)$\\
$\clustering := (\clustering \setminus \{C_1, C_2\}) \cup \{C_1 \cup C_2\}$.
}
}
\end{algorithm}

\begin{algorithm}[t]
\caption{Average linkage HAC---graph contraction version.}\label{alg:contractions}
\KwIn{$G = (V, E, w)$}
\SetKwFunction{FSim}{Similarity}
\SetKwFunction{FHac}{HAC}
\SetKwProg{Fn}{Function}{:}{}

\Fn{\FSim{$x, y, w, S$}}{
\Return $w(x,y) / (S(x) \cdot S(y))$
}

\Fn{\FHac{$G$}}{
$S := $ a function mapping each element of $V$ to 1\\
\While{$\exists_{xy \in E}$ s.t. $\FSim(x, y, w, S) > 0$}{
$xy = \argmax_{xy \in E} \FSim(x, y, w, S)$\\
Contract $x$ with $y$ in G creating a vertex $z$. The parallel edges that are created are merged into a single edge whose weight is the sum of the merged edge weights. Any resulting self-loops are removed.\\
Set $S(z) := S(x) + S(y)$
}
}
\end{algorithm}

\section{An \texorpdfstring{\boldmath$\Omega(n^{3/2 - \epsilon})$}{Omega(n\^(3/2-epsilon))} Conditional Lower Bound for Average Linkage HAC}\label{sec:worklb}

In this section, we show an $\Omega(n^{3/2 - \epsilon})$ conditional lower bound on the time required to solve average linkage HAC on general weighted graphs. Specifically, we show this lower bound assuming the Combinatorial Boolean Matrix Multiplication (BMM) conjecture, a central conjecture in fine-grained complexity about the time required to multiply two $n \times n$ boolean matrices~\cite{williams2010subcubic, abboud2024faster}.

\begin{conjecture}[Combinatorial BMM]\label{conj:bmm}
Combinatorial algorithms cannot solve Boolean Matrix Multiplication in time $O(n^{3-\epsilon})$ for $\epsilon > 0$.
\end{conjecture}
\noindent We refer to~\cite{abboud2024faster} for an in-depth discussion of the somewhat informal notion of ``combinatorial'' algorithms and more on Conjecture~\ref{conj:bmm} and its history.

In this work we will make use of an equivalent characterization of the BMM conjecture due to \cite{williams2010subcubic}. Specifically, \cite{williams2010subcubic} shows that the BMM problem is sub-cubically equivalent to the Triangle Detection problem: the problem of deciding whether or not an input graph $G$ contains a triangle (i.e., cycle of $3$ vertices). The following summarizes this result.
\begin{theorem}[Theorem 1.3 of \cite{williams2010subcubic}]\label{thm:bmmeq}
    Combinatorial algorithms cannot solve Triangle Detection in time $O(n^{3-\epsilon})$ for  $\epsilon > 0$ unless \Cref{conj:bmm} is false.
\end{theorem}
Thus, we give a reduction from Triangle Detection to average linkage HAC. Our reduction will quadratically increase the number of vertices of the input Triangle Detection instance, and therefore give an $\Omega(n^{3/2 - \epsilon})$ lower bound for average linkage HAC. In the rest of this section, we show the following quadratic-blowup reduction from Triangle Detection to average linkage HAC.

\begin{theorem}\label{thm:lbreduction}
Given a Triangle Detection instance on graph $G$ with $t$ vertices and $m$ edges, there is a reduction that runs in $O(t^2)$ time and constructs an instance of average linkage HAC on graph $G'$ with $t + t^2$ vertices and $t^2 + m$ edges. 
Furthermore, given the sequence of merges performed by average linkage HAC on $G'$, we can solve Triangle Detection on $G$ in time $O(t^2)$.
\end{theorem}

As a corollary of this reduction and \Cref{thm:bmmeq}, we obtain the following conditional lower-bound on the running time of HAC.

\mainLowerSeq*

As a second corollary, we obtain a conditional lower-bound in terms of the optimal running time of matrix multiplication for two $n \times n$ binary matrices. 
Matrix multiplication can be solved in time $O(n^{\omega})$ where $2 \leq \omega < 2.3716$~\cite{williams2024new}. 
An extensive line of research on matrix multiplication over the past thirty years has only improved $\omega$ from $2.376$ to $2.3716$, with the current state-of-the-art being due to a very recent result of Williams et al.~\cite{williams2024new} (for a subset of the historical advances in this area see, e.g., \cite{coppersmith1982asymptotic, le2012faster, alman2021refined, williams2012multiplying}).
The fastest known algorithm for triangle detection works by simply reducing the problem to matrix multiplication and therefore runs in $O(n^{\omega})$ time.
Surprisingly, despite triangle detection only returning a single bit (whether a triangle exists or not in $G$), the problem can be used to give a sub-cubic reduction for boolean matrix multiplication (where the output is $n^2$ bits).
In particular, an algorithm for triangle detection running in time $O(n^{3-\delta})$ for some $\delta > 0$ yields an algorithm for matrix multiplication in time $O(n^{3-\delta/3})$~\cite{williams2010subcubic}.
Using this fact, we can derive a conditional lower bound based on the value of $\omega$.

\workmatmul*


An interesting open question is whether there are faster non-combinatorial algorithms that can leverage fast matrix multiplication or Strassen-like techniques and improve over the $\Omega(n^{3/2 - \epsilon})$ barrier for combinatorial algorithms for average linkage HAC.

\subsection{Reduction}
We now prove Theorem~\ref{thm:lbreduction} by
giving a quadratic-time reduction from triangle detection to average linkage HAC.
The reduction is loosely inspired by a recent lower-bound result for multidimensional range queries~\cite{lau2021algorithms}.
The input to the reduction is an unweighted graph $G$ on $t$ vertices with $m$ edges; the problem is to detect whether $G$ has a triangle.
To do this, we will construct a HAC instance on an edge-weighted graph $G'$ with $t + t^2$ vertices and $t^2 + m$ edges. 
We will show that the specific way in which an exact HAC algorithm merges the edges in this instance reveals whether or not $G$ has a triangle.

\paragraph{Constructing $G'$}
Let $N_{G}(v)$ denote the neighbors of a vertex $v \in G$ (note that $v \notin N_{G}(v)$).
We define $G'$ as follows. 
We start by adding all vertices and edges from $G$, that is the $t$ vertices $v_1, \ldots, v_t$ from $G$, including all of their incident edges $N_{G}(v_i)$.
We call these the {\em core} vertices. 
The initail weight of the edges between any two core vertices is set to $1$.

In addition to the core vertices, we add an additional $t^2$ {\em leaf} vertices that we connect to the core vertices with specific edge weights.
We add the $t^2$ leaf vertices over a sequence of $t$ {\em rounds} where the $i$-th round connects one new leaf vertex to every core vertex.
The weights to the newly added leaves depend on the neighbors of the node $v_i$ in the original graph $G$, and are set as follows:
\begin{enumerate}[label=(\textbf{\arabic*})]
\item A core vertex $v_j$ is connected to its new leaf with an edge of weight $(1/i) - \epsilon$ if $v_j \in N_{G}(v_i)$.\label{typeone}
\item A core vertex $v_j$ is connected to its new leaf with an edge of weight $(1/i) + \epsilon$ if $v_j \notin N_{G}(v_i)$.\label{typetwo}
\end{enumerate}

\noindent See \Cref{fig:HACSeqRed} for an illustration of our reduction.
\begin{figure}[h]
    \centering
    \begin{subfigure}[b]{0.25\textwidth}
        \centering
        \includegraphics[width=\textwidth,trim=0mm 0mm 370mm 160mm, clip]{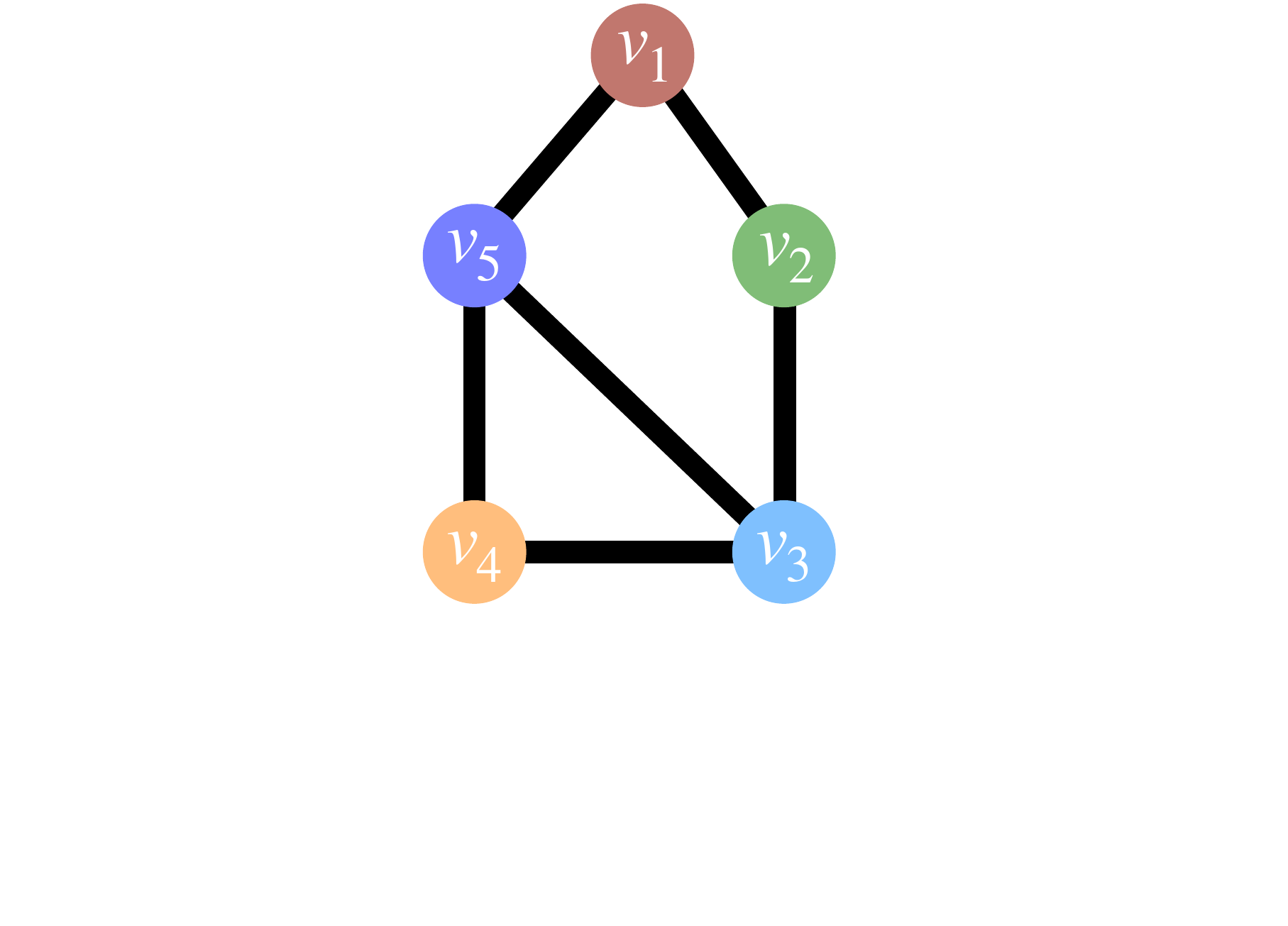}
        \caption{Input $G$.}\label{sfig:CBMMRed1}
    \end{subfigure} \hfill
       \begin{subfigure}[b]{0.25\textwidth}
        \centering
        \includegraphics[width=\textwidth,trim=0mm 0mm 370mm 80mm, clip]{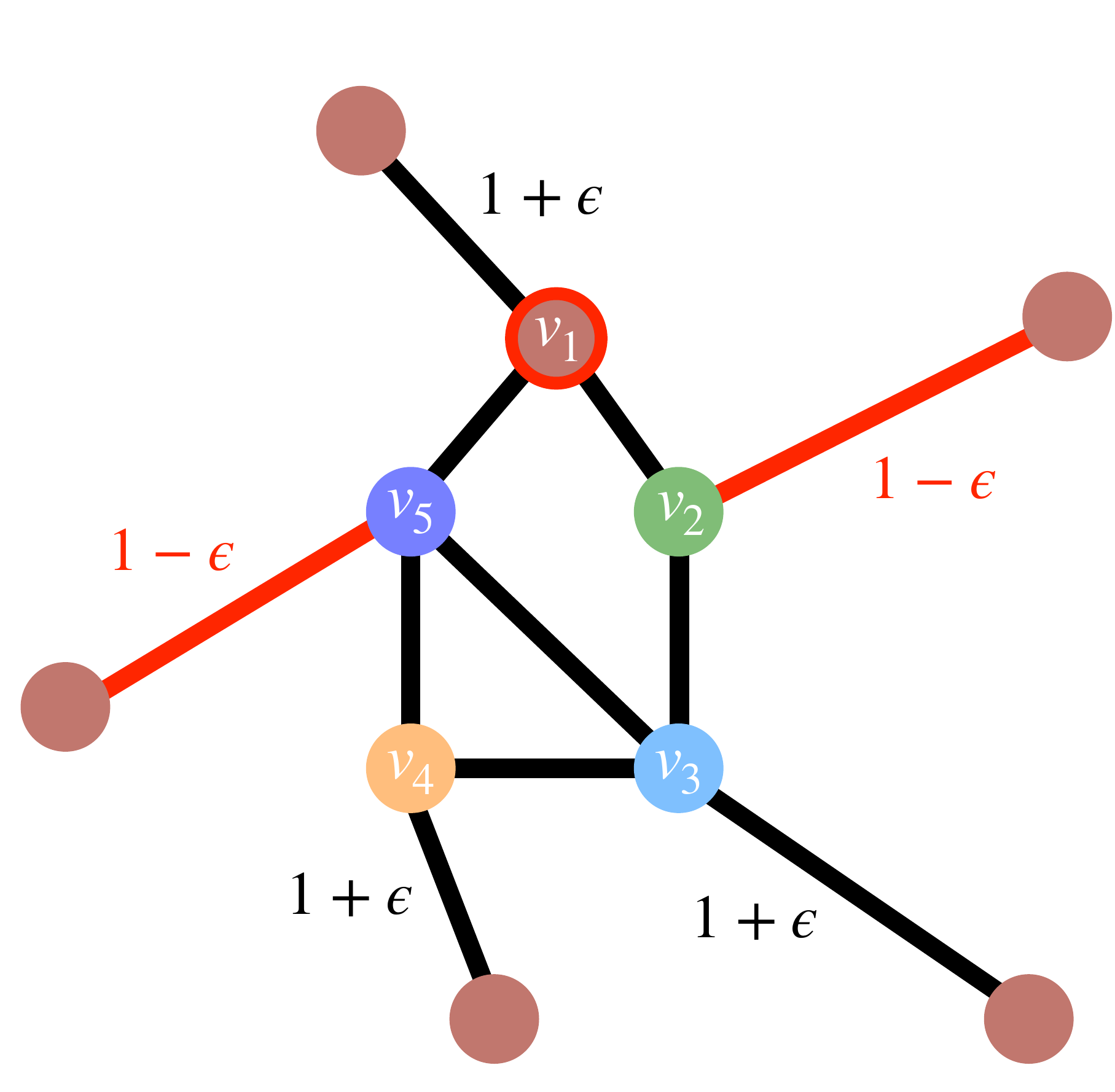}
        \caption{Round $1$.}\label{sfig:CBMMRed2}
    \end{subfigure} \hfill
       \begin{subfigure}[b]{0.25\textwidth}
        \centering
        \includegraphics[width=\textwidth,trim=0mm 0mm 370mm 80mm, clip]{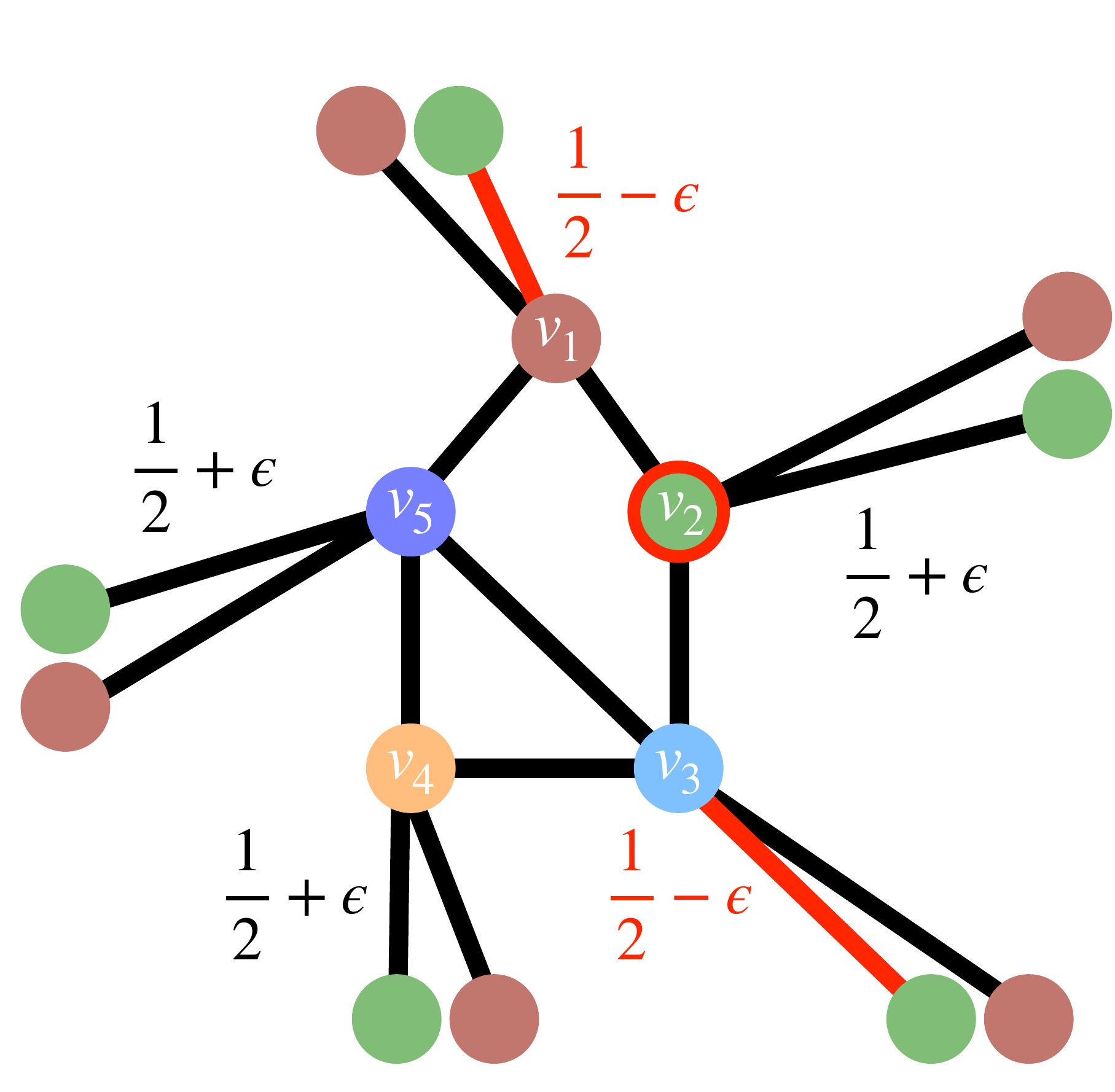}
        \caption{Round $2$.}\label{sfig:CBMMRed3}
    \end{subfigure} \hfill
    \begin{subfigure}[b]{0.25\textwidth}
        \centering
        \includegraphics[width=\textwidth,trim=0mm 0mm 370mm 80mm, clip]{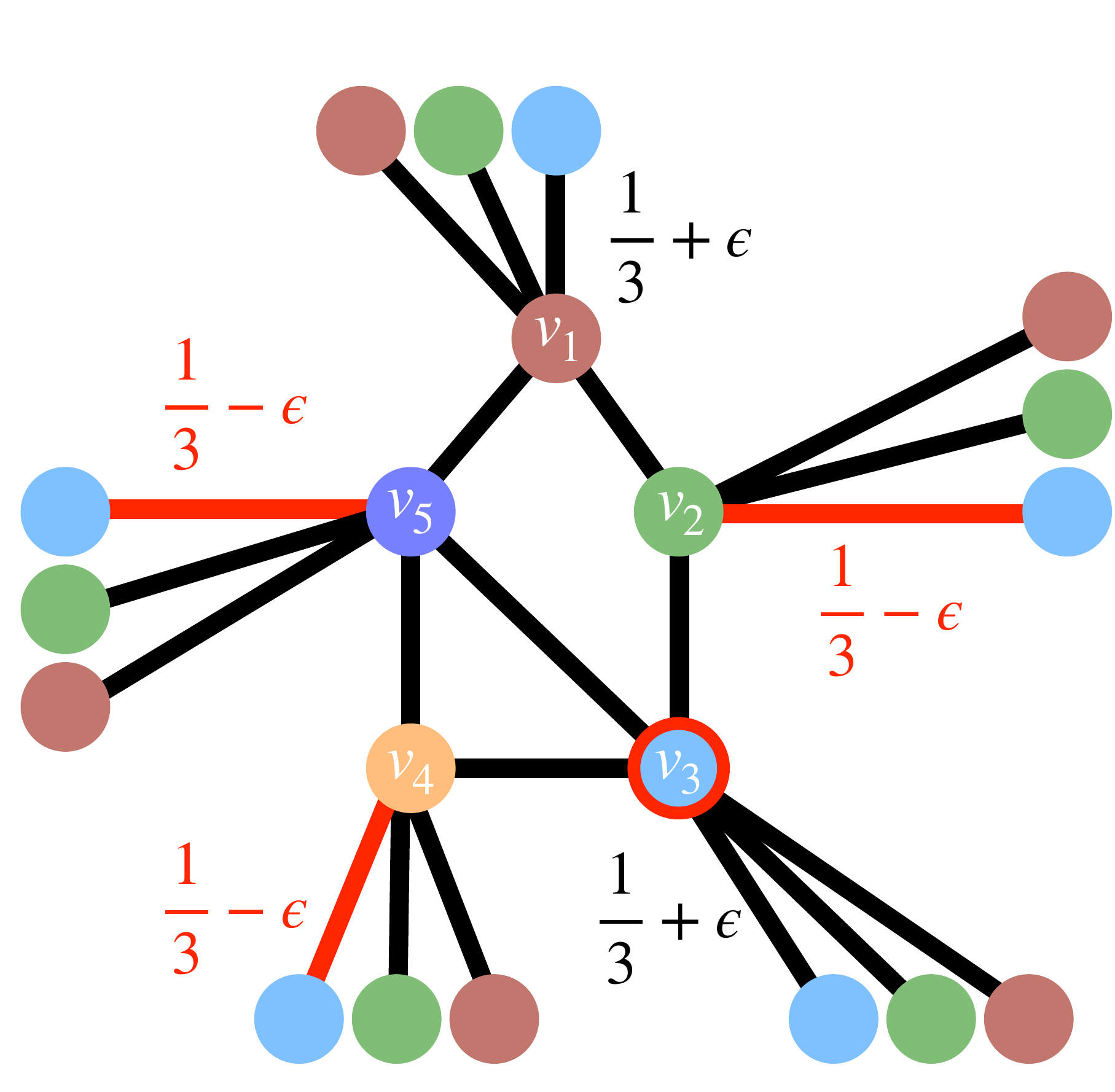}
        \caption{Round $3$.}\label{sfig:CBMMRed4}
    \end{subfigure} \hfill
    \begin{subfigure}[b]{0.25\textwidth}
        \centering
        \includegraphics[width=\textwidth,trim=0mm 0mm 370mm 80mm, clip]{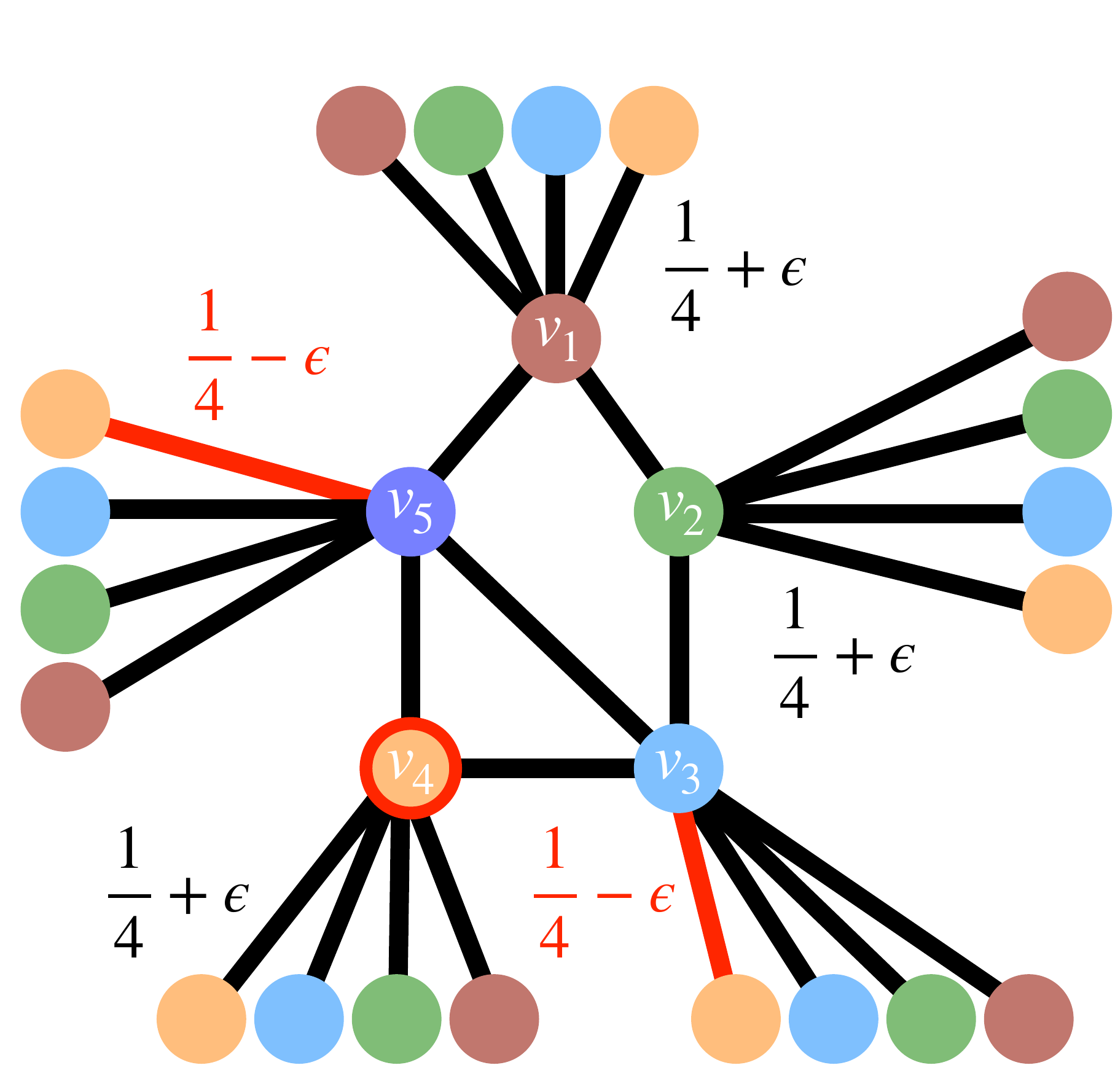}
        \caption{Round $4$.}\label{sfig:CBMMRed5}
    \end{subfigure} \hfill
    \begin{subfigure}[b]{0.25\textwidth}
        \centering
        \includegraphics[width=\textwidth,trim=0mm 0mm 370mm 80mm, clip]{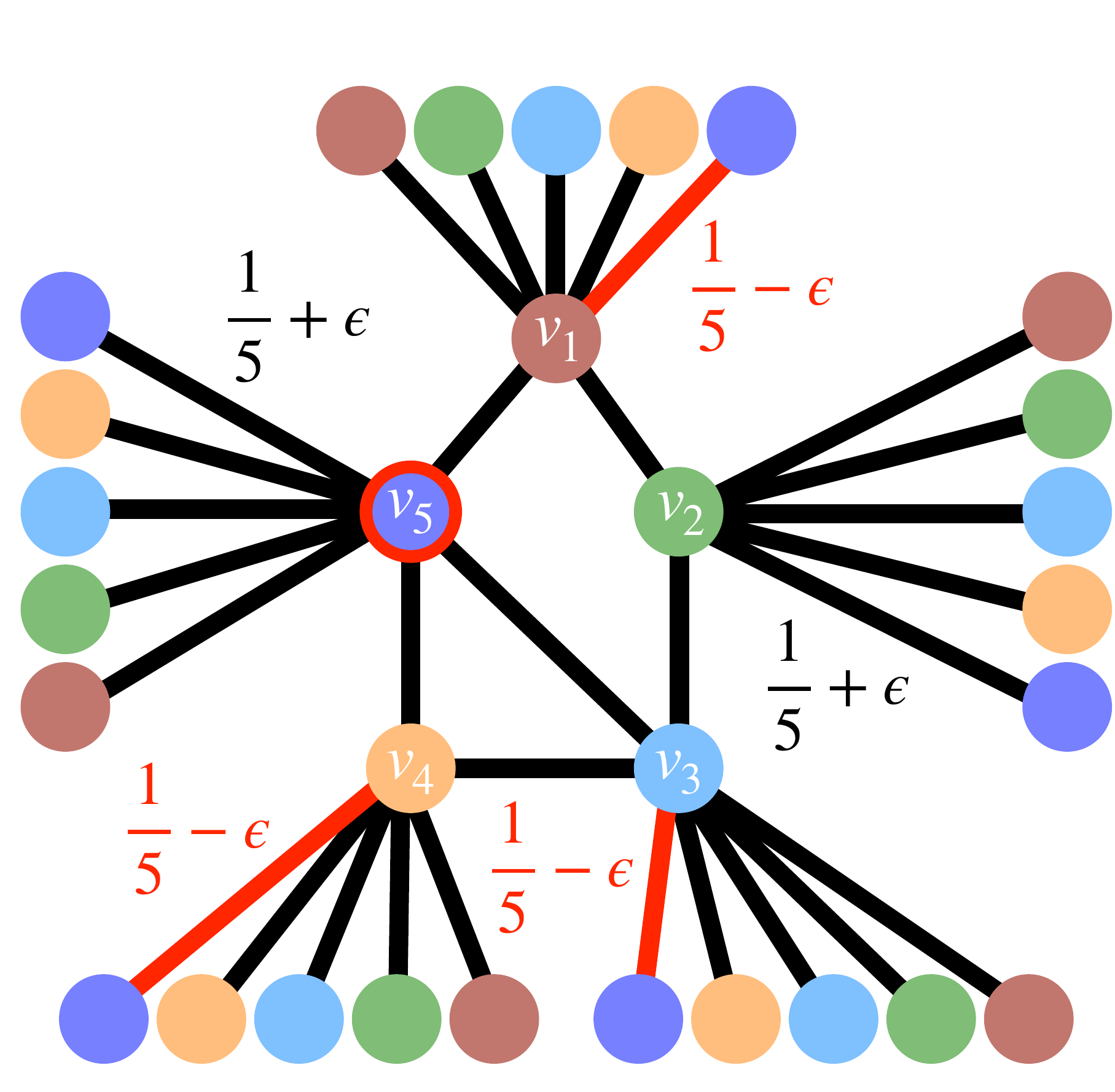}
        \caption{Round $5$.}\label{sfig:CBMMRed6}
    \end{subfigure} \hfill
    \caption{Our triangle detection reduction where we compute $G'$ from $G$ by adding $t=5$ nodes over $t$ rounds. \ref{sfig:CBMMRed6} gives $G'$. Each node labeled according to its round and corresponding vertex in $G$. Edges labelled with their weight in the round they are added (edges of $G$ have weight $1$). For the $i$th round we highlight in red $v_i$ and the edges added with weight $1/i-\epsilon$.} \label{fig:HACSeqRed}
\end{figure}

\paragraph{Running HAC on $G'$}
Having defined $G'$, let us consider the merges that the exact HAC algorithm will make on this instance.
In this section, for succinctness we use {\em weight} to refer to the normalized (i.e., average linkage) weight.
HAC will begin by merging the maximum weight edges.
The maximum weight initially depends on the structure of $N_{G}(v_1)$.
First, all core vertices that are not in $N_{G}(v_1)$ will merge with their round 1 leaves (since these edges have weight $1+\epsilon$) increasing their cluster size to $2$.
This leaves all core vertices in $N_{G}(v_1)$. Any edge in $G'$ between two such core vertices will have weight $1$, and will be merged next.
Crucially, any merge in this round with a weight of $1$ indicates a triangle incident on $v_1$ since the two core vertex endpoints of the edge must be contained in $N_{G}(v_1)$, hence connected by 
edges to $v_1$.
If no edges of weight $1$ merge, the remaining leaves that we added to core vertices in $N_{G}(v_1)$ merge into their neighboring core vertex.

Assuming we did not merge any weight $1$ edges in the first round, at this point the cluster size of each core vertex is $2$, and so the weight of any edge originally in $G$ (between two core vertices) will be $1/4$.
The edge weights to leaves in round $2$ will be $((1/2) \pm \epsilon)/2 = 1/4 \pm (\epsilon / 2)$, which is larger than $1/4$ for edges of Type~\ref{typetwo}.
Therefore, the same argument for how edges merge in round $1$ can be inductively applied to the next round.
The edge weights in round $i$ for edges between core vertices will be $1/i^2$, and by the same argument as before, the Type~\ref{typetwo} (Type~\ref{typeone}) edges will be larger (smaller) by $\epsilon/i$.
As a result of how an exact average linkage HAC will merge the edges of $G'$, we obtain the following lemma:

\begin{lemma}
Consider the sequence of merges performed by the HAC algorithm on $G'$. If the merge sequence consists of $t^2$ merges, which first merge all leaf vertices, and only then makes merges between core vertices, then $G$ does not contain any triangles.
If the merge sequence merges any edge between two core vertices in the first $t^2$ merges, then $G$ contains a triangle.
\end{lemma}

\paragraph{Completing the Reduction}
We will now complete the proof of Theorem~\ref{thm:lbreduction}.
Suppose we are given an instance of \textsc{Triangle Detection} on $n$ vertices.
Conjecture~\ref{conj:bmm} implies that this instance cannot be solved by combinatorial algorithms in $O(n^{3-\epsilon})$ time for any $\epsilon > 0$.

Let the time complexity of HAC on a graph $G$ with $n$ vertices and $m$ edges be $T_{\mathsf{HAC}}(n,m)$.
Suppose HAC can be solved combinatorially in $O(n^{3/2 - \epsilon})$ time.
Given a \textsc{Triangle Detection} instance on $n$ vertices we create a graph $G'$ with $O(n^2)$ vertices and $O(n^2)$ edges, and run HAC on $G'$.
The running time of the reduction is $O(n^2)$, and the running time of HAC on $G'$ is $O((n^{2\cdot (3/2-\epsilon)}) = O(n^{3 - 2\epsilon})$, which will falsify Conjecture~\ref{conj:bmm} by \Cref{thm:bmmeq}.
Thus, conditional on Conjecture~\ref{conj:bmm}, there is no algorithm for HAC running in time $T_{\mathsf{HAC}}(n, m) = O(n^{3/2 - \epsilon})$ for any constant $\epsilon > 0$, completing the proof of Theorems~\ref{thm:lbreduction} and \ref{thm:worklb}. 
The same argument, under the assumption that triangle detection cannot be solved in $O(n^{\omega - \epsilon})$ time for any constant $\epsilon > 0$ implies Theorem~\ref{cor:workmatmul}.

\section{Average Linkage HAC is Hard to Parallelize Even on Trees}
In this section we prove that average linkage HAC is likely hard to parallelize by showing it is \CC-hard even on low depth trees.
%
%
We begin with some preliminaries. The formal definition of \CC-hardness we will use is as follows.
\begin{definition}[\CC-Hard]\label{dfn:CCComplete}
    A problem is \CC-hard if all problems of \CC are logspace-reducible to it.
\end{definition}

For our purposes we will not need to define the class \CC. Rather, we only need the above definition of \CC-hardness and a single \CC-hard problem, \lfm. Recall that a matching of a graph $G = (V, E)$ is a subset of edges $M \subseteq E$ if each vertex is incident to at most one edge of $M$. A matching is said to be maximal if each $e = \{u, v\} \not \in M$ satisfies the property that either $u$ or $v$ is incident to an edge of $M$. The greedy algorithm for maximal matching initializes $M$ as $\emptyset$ and then simply iterates over the edges of $E$ in some order and adds the current edge $e$ to $M$ if the result of doing so is a matching.

\begin{problem}[\lfm]\label{dfn:LFMM} An instance of lexicographically first maximal matching (LFM Matching) consists of a bipartite graph $G=(V = L \sqcup R, E)$ with vertices ordered as $L = (l_0, l_1, \ldots, l_{n-1})$ and $R = (r_0, r_1, \ldots, r_{n-1})$. The lexicographically first maximal matching is the matching obtained by running the greedy algorithm for maximal matching on edges ordered first by their endpoint in $L$ and then by their endpoint in $R$. That is, in this ordering $e = \{l_i, r_j\}$ precedes $e' = \{l_{i'}, r_{j'}\}$ iff (1) $i < i'$ or (2) $i = i'$ and $j < j'$. Our goal is to decide if a designated input edge is in the LFM Matching.
\end{problem}
The following summarizes known hardness of \lfm.
\begin{theorem}[\cite{subramanian1989new,mayr1992complexity}]\label{thm:LFMCCComplete}
    \lfm is \CC-hard.
\end{theorem}

Next, we introduce the search variant of the HAC problem whose \CC-hardness we will prove. 

\begin{problem}[\hac]\label{prob:hac}
An instance of \hac consists an undirected graph $G = (V,E)$, along with edge weights $w:E \to \mathbb{R}_{\geq 0}$. Consider the sequence $(C_1,C_1'),(C_2,C_2'),\dots$ of cluster merges produced by the procedure \FHac{$G$} from Algorithm \ref{alg:staticgraph}. Given any pair of vertices $u,v \in G$, the goal of the \hac problem is to output the index $i$ such that $u,v$ first merge together at step $i$, namely $u \in C_i$ and $v \in C_i'$ (or $u \in C_i'$ and $v \in C_i$). 
\end{problem}

To prove the hardness of \hac, we will first prove the hardness of an intermediate problem, called \am. The construction of the \am problem will be more amenable to our reductions, and therefore simplify the following exposition. See \Cref{fig:aMin} for an illustration of \am.

\begin{problem}[\am]\label{l:am}
An instance of \am consists of a ($0$-based indexed) $n \times n$ matrix $A$ where each row contains a permutation of $\{0, \ldots, n-1\}$ and some index $x \in [0, n)$.
The goal is to simulate the following algorithm.
Start with $I = \{0, \ldots, n-1\}$ and execute the following steps for $i = 0, \ldots, x$:
\begin{enumerate}
    \item Let $k_i = \argmin_{j \in I} A[i, j]$.
    \item Set $I := I \setminus \{k_i\}$. 
\end{enumerate}
Our goal is to compute $k_x$.
\end{problem}

\begin{figure}[h]
    \centering
    \begin{subfigure}[b]{0.15\textwidth}
        \centering
        \includegraphics[width=\textwidth,trim=0mm 0mm 250mm 0mm, clip]{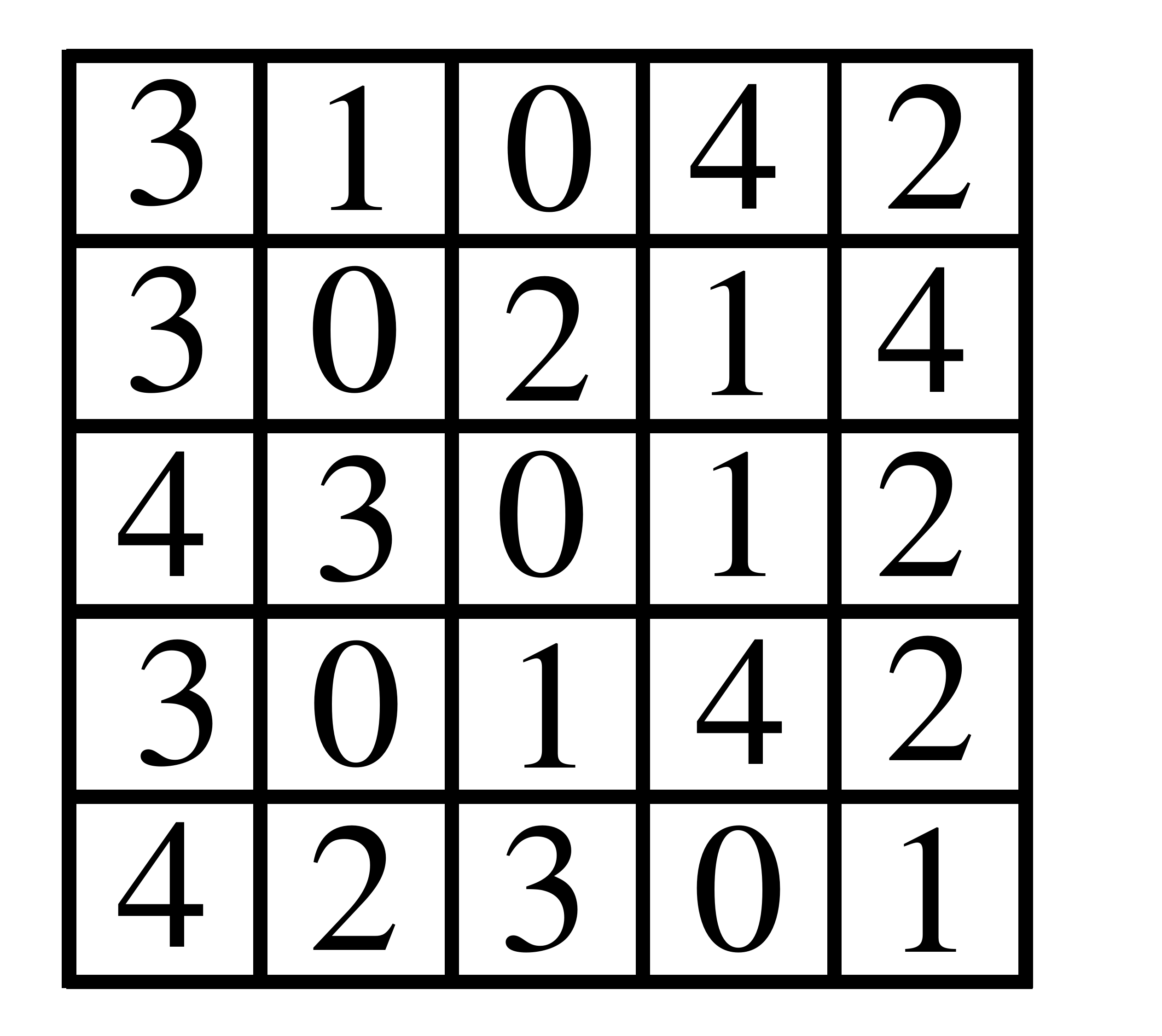}
        \caption{Input $A$.}\label{sfig:aMin1}
    \end{subfigure} \hfill
       \begin{subfigure}[b]{0.15\textwidth}
        \centering
        \includegraphics[width=\textwidth,trim=0mm 0mm 250mm 0mm, clip]{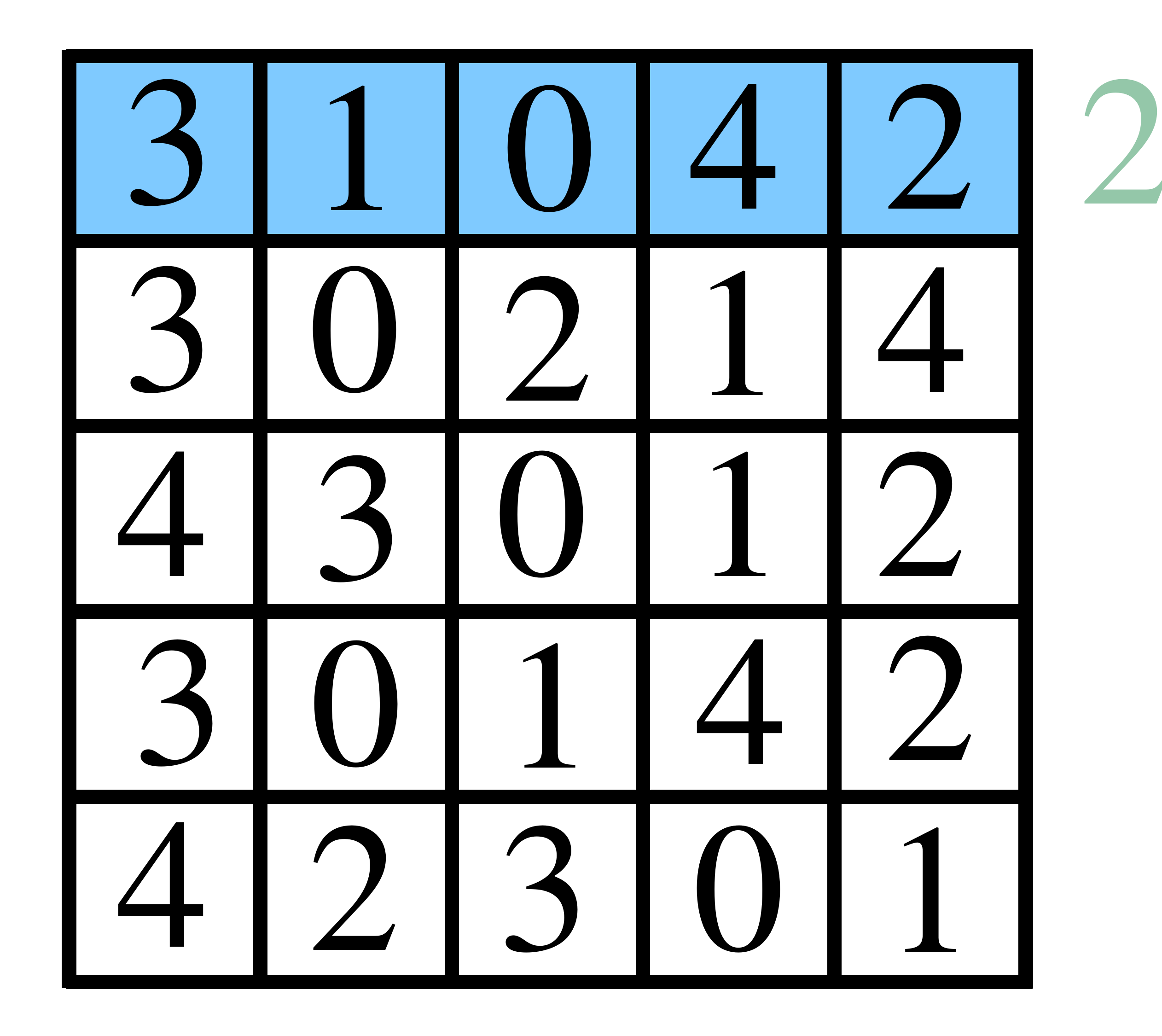}
        \caption{Step $1$.}\label{sfig:aMin2}
    \end{subfigure} \hfill
       \begin{subfigure}[b]{0.15\textwidth}
        \centering
        \includegraphics[width=\textwidth,trim=0mm 0mm 250mm 0mm, clip]{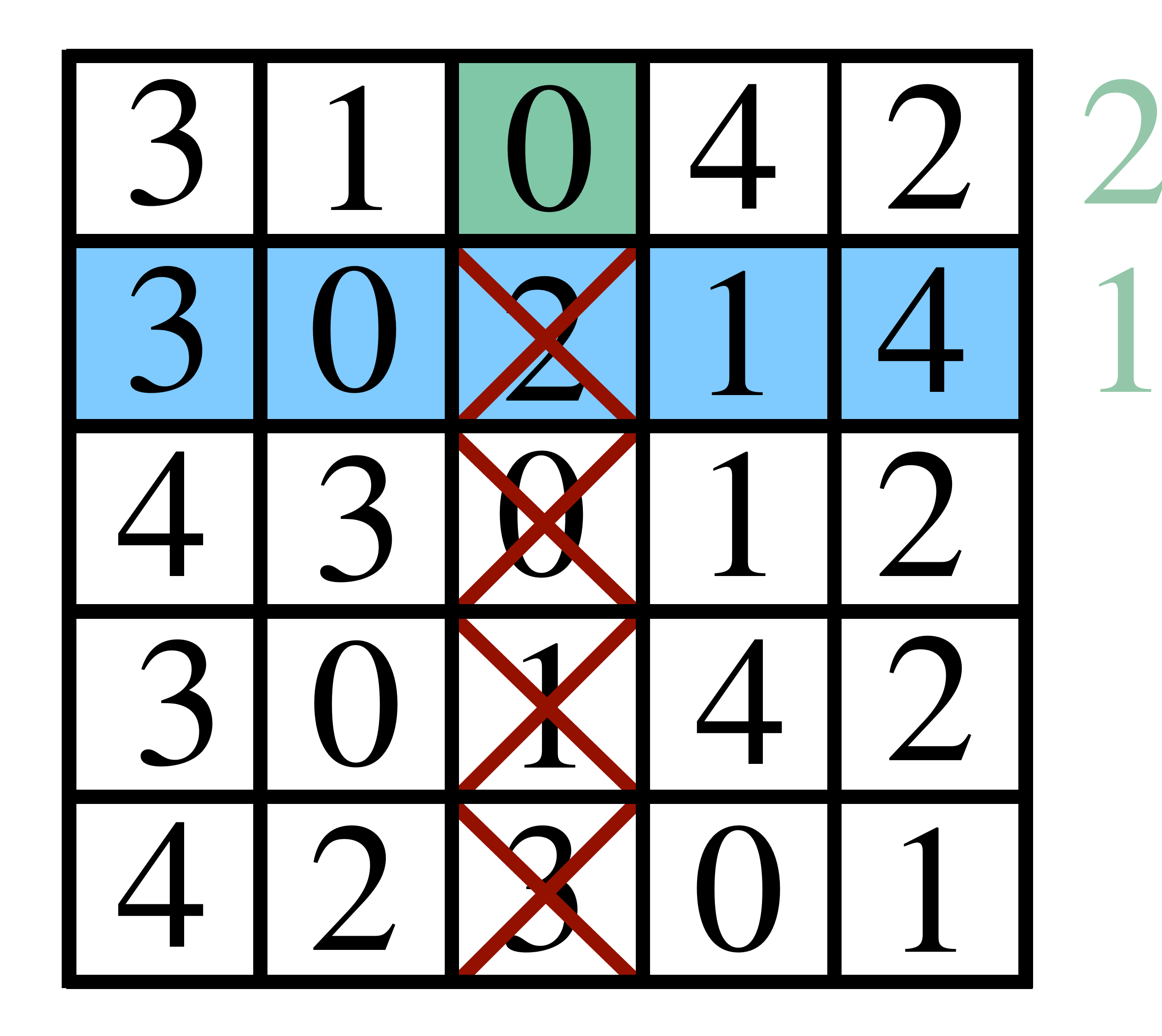}
        \caption{Step $2$.}\label{sfig:aMin3}
    \end{subfigure} \hfill
    \begin{subfigure}[b]{0.15\textwidth}
        \centering
        \includegraphics[width=\textwidth,trim=0mm 0mm 250mm 0mm, clip]{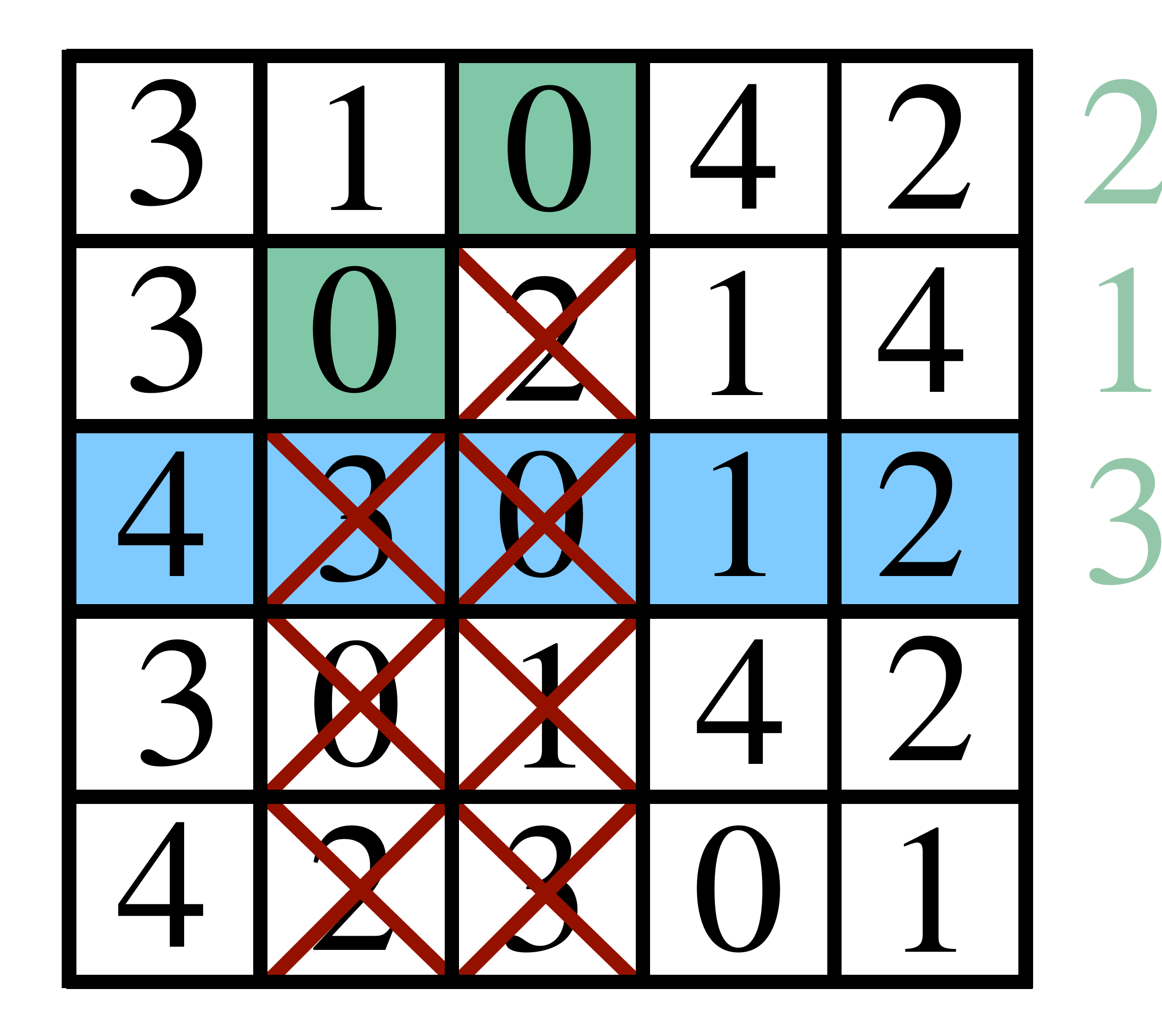}
        \caption{Step $3$.}\label{sfig:aMin4}
    \end{subfigure} \hfill
    \begin{subfigure}[b]{0.15\textwidth}
        \centering
        \includegraphics[width=\textwidth,trim=0mm 0mm 250mm 0mm, clip]{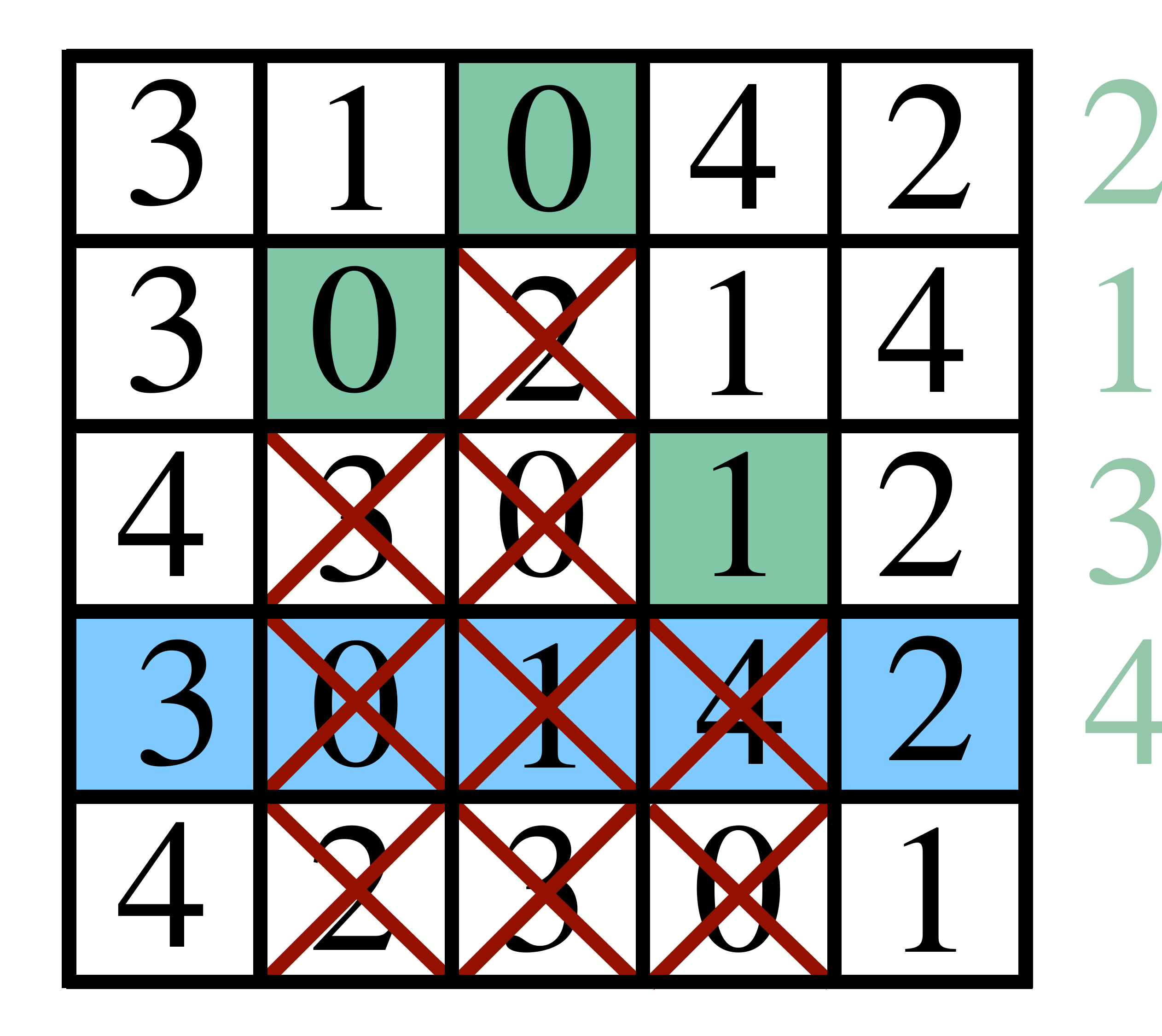}
        \caption{Step $4$.}\label{sfig:aMin5}
    \end{subfigure} \hfill
    \begin{subfigure}[b]{0.15\textwidth}
        \centering
        \includegraphics[width=\textwidth,trim=0mm 0mm 250mm 0mm, clip]{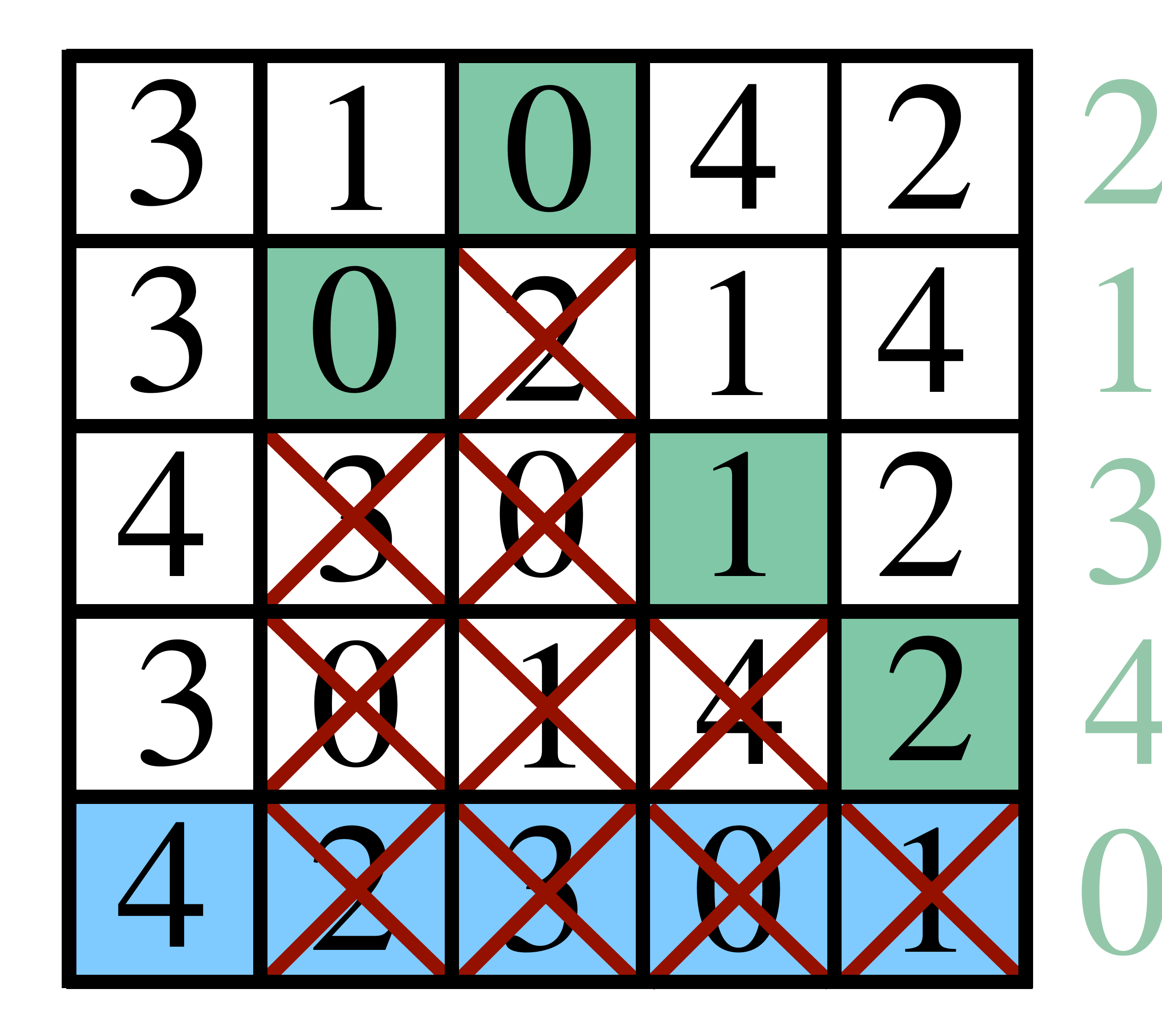}
        \caption{Step $5$.}\label{sfig:aMin6}
    \end{subfigure} \hfill
    \caption{\am on matrix $A$. The row considered in each step is shown in blue. 
    $k_i$ for the $i$-th row written to the right of $A$ in green with witnessing entry of $A$ also in green. Indices removed from $I$ in relevant rows crossed out in red.} \label{fig:aMin}
\end{figure}

Observe that both problems \ref{prob:hac} and \ref{l:am} are defined as having an algorithm output an index $i \in \{1,2,\dots\}$. Thus, these can be considered search problems instead of decision problems. We choose to work with the search versions of these problems for simplicity of our reductions, however, our reduction naturally extends to the decision variants  (e.g., where the algorithm is given $u,v \in V$ and an index $i$ and asked if $u,v$ merge on step $i$). 

We first prove the \CC-hardness of this intermediate problem. See \Cref{fig:aMinRed} for an illustration of our reduction from \lfm to \am.
\begin{lemma}
\am is \CC-hard.
\end{lemma}
\begin{proof}
    By \Cref{thm:LFMCCComplete} and the definition of \CC-hardness (\Cref{dfn:CCComplete}), it suffices to argue that \lfm (\Cref{dfn:LFMM}) is logspace reducible to \am (\Cref{l:am}).
    We begin by describing our reduction and then observe that it only requires logarithmic space. The basic idea of the reduction is to associate with each vertex on the left side of our instance of \lfm a row of the matrix of \am and each vertex on the right side of \lfm a column of the matrix of \am.

    More formally, consider an instance of \lfm on graph $G = (V = L \sqcup R, E)$ where our goal is to decide if a given edge $e$ is in the LFM matching. 
    We consider the following instance of \am to solve this on a $2n \times 2n$ size matrix $A$. We will refer to an index $i$ as a \emph{dummy index} if $i \geq n$. We will construct the $i$th row of $A$ to correspond to a permutation that first gives the indices of all neighbors of $l_i$ in $R$ sorted according to the ordering of $R$ then gives all dummy indices then gives the indices of non-neighbors of $l_i$ in $R$. 
    
    More formally, for a vertex $l_i \in L$, we let $R_i \subseteq R$ be the vertices of $G$ connected to $l_i$ and we let $\bar{R_i} = R \setminus R_i$ be all vertices not adjacent to $l_i$. Furthermore, for $r_j \in R_i$, we let $\text{ind}(r_j) \in [|R_i|]$ be the index of $r_j$ in $R_i$ when we sort the vertices of $R_i$ according to the ordering on $R$. Note that, in general, it need not be the case that $\text{ind}(r_j) = j$. Symmetrically, for a vertex $r_j \in \bar{R_i}$, we let $\text{ind}(r_j) \in [|\bar{R_i}|]$ be the index of $r_j$ when we sort $\bar{R_i}$ by the ordering on $R$. Then, we define $A$ as follows.
    \begin{align*}
        A[i,j] := \begin{cases} \text{ind}(r_j) &\text{if $j < n$ and $r_j \in R_i$}\\
        j - n + |R_i| & \text{if $j \geq n$}\\
        \text{ind}(r_j) + n + |R_i|  & \text{if $j < n$ and $r_j \not \in R_i$}
        \end{cases}
    \end{align*}

    Lastly, let $x$ (the index for which we would like to compute $k_x$ in our instance of \am) be the index of the endpoint of $e$ in $L$.
    Once \am computes $k_x$, we verify whether or not it corresponds to the endpoint of $e$ in $R$ to determine the final output of the \lfm instance. Again, see \Cref{fig:aMinRed}.

    We now argue correctness of the reduction. A straightforward proof by induction on $i$ demonstrates that at the beginning of the $i$th round of the \am algorithm we have that $I$ consists of at least $n-i$ dummy indices and $j < n$ is not in $I$ only if $r_j$ is in the \lfm and is matched to some $l_{i'}$ for $i' < i$. It follows that $e = (l_i, r_j)$ is in the \lfm iff $k_i = j$, showing correctness of our reduction. 

    It remains to show that the above reduction can be done with logspace. In order to do so, we must argue that $A[i,j]$ can be computed with logspace for every $i$ and $j$. Doing so is trivial given the above definition of $A[i,j]$. 
\end{proof}

\begin{figure}[h]
    \centering
    \begin{subfigure}[b]{0.24\textwidth}
        \centering
    \includegraphics[width=\textwidth,trim=0mm 50mm 350mm 0mm, clip]{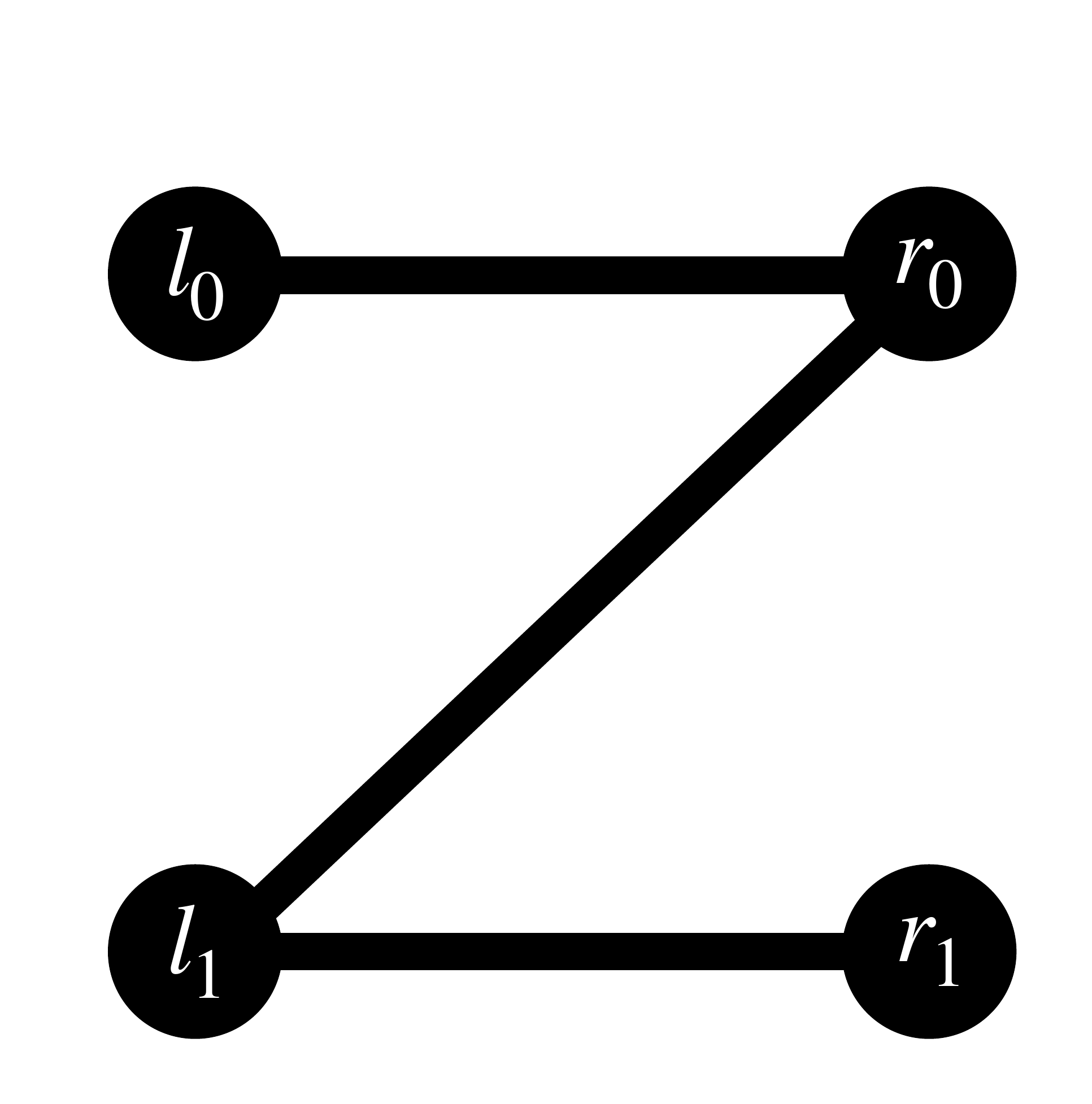}
        \caption{\lfm.}\label{sfig:aMinRed1}
    \end{subfigure} \hfill
       \begin{subfigure}[b]{0.24\textwidth}
        \centering
    \includegraphics[width=\textwidth,trim=0mm 50mm 350mm 0mm, clip]{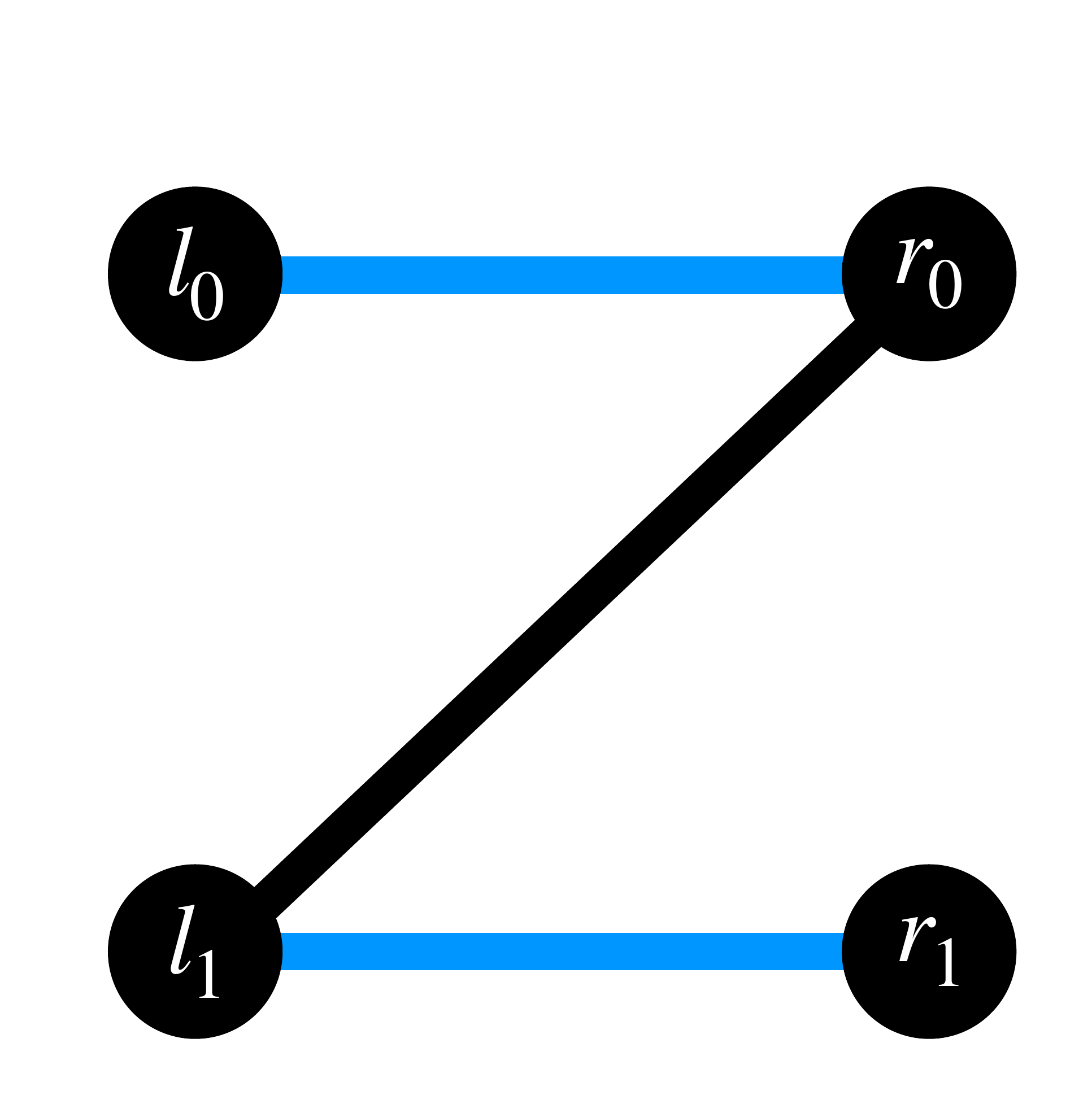}
        \caption{Solution.}\label{sfig:aMinRed2}
    \end{subfigure} \hfill
       \begin{subfigure}[b]{0.24\textwidth}
        \centering
    \includegraphics[width=\textwidth,trim=0mm 0mm 230mm 0mm, clip]{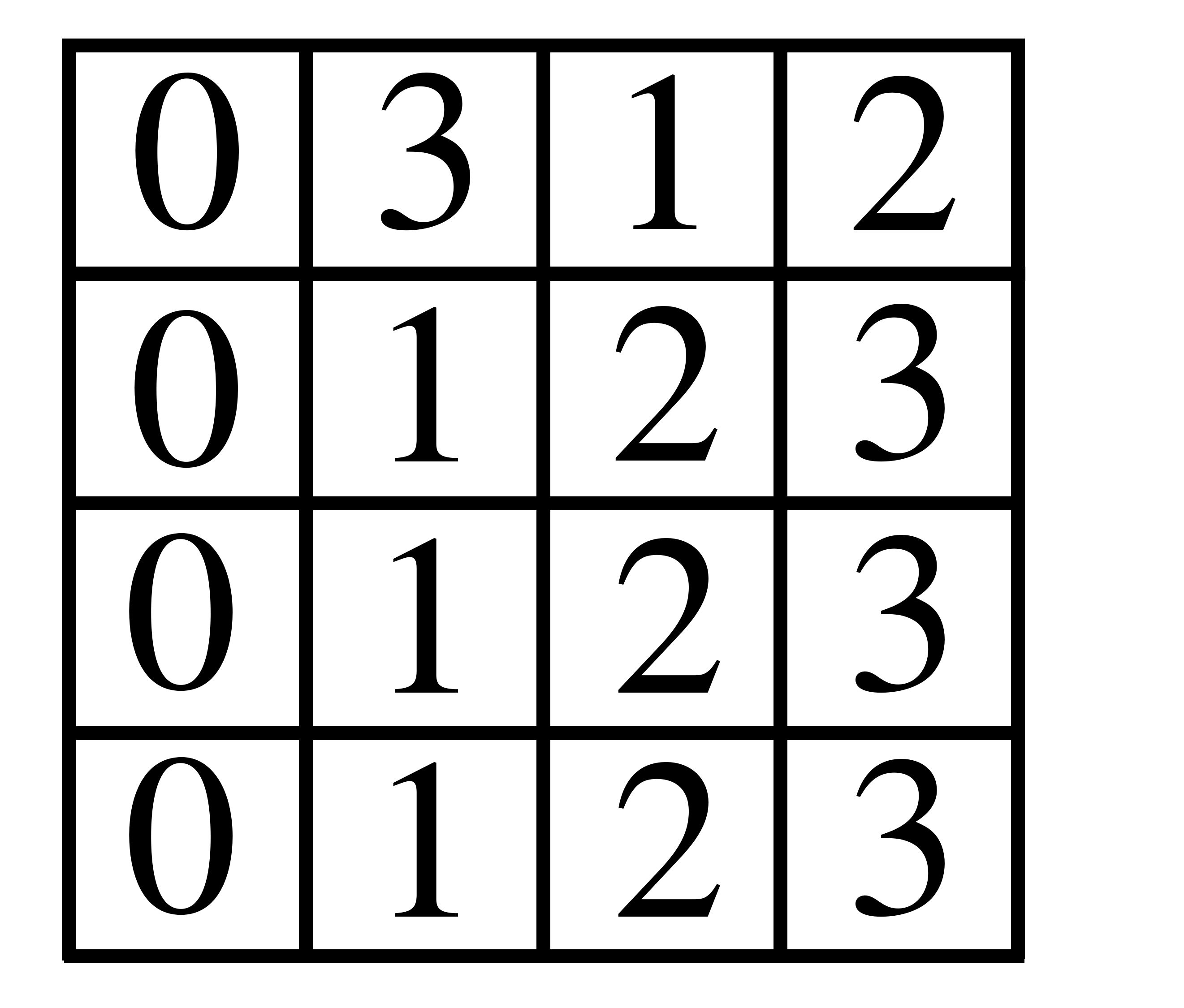}
        \caption{\am.}\label{sfig:aMinRed3}
    \end{subfigure} \hfill
    \begin{subfigure}[b]{0.24\textwidth}
        \centering
        \includegraphics[width=\textwidth,trim=0mm 0mm 230mm 0mm, clip]{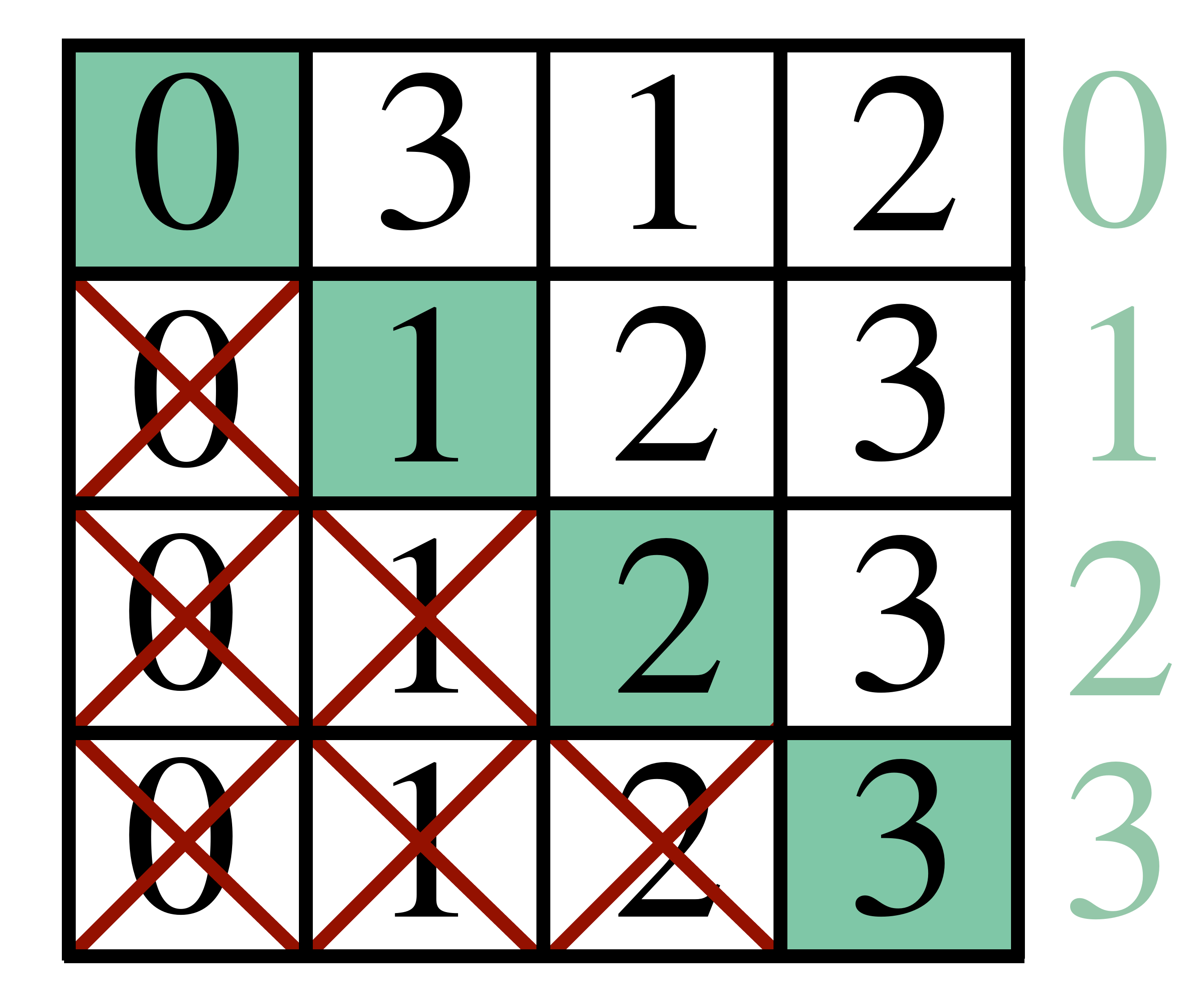}
        \caption{Solution.}\label{sfig:aMinRed4}
    \end{subfigure} \hfill
    \caption{Reduction from \lfm to \am. \ref{sfig:aMinRed1} gives the \lfm instance and \ref{sfig:aMinRed2} its solution. \ref{sfig:aMinRed3} gives the \am instance from the reduction and \ref{sfig:aMinRed4} its solution.} \label{fig:aMinRed}
\end{figure}

Concluding, we use the \CC-hardness of \am to prove the \CC-hardness of \hac. 
\mainLowerPar*

\begin{proof}
Our reduction shows how to reduce an instance of \am (\cref{l:am}) of size $n$ to an instance of average linkage HAC on a tree.
We build a rooted tree, in which each root-to-leaf path has length $2$ (i.e., the tree has depth $2$). 
We call the neighbors of the root \emph{internal} nodes.
Observe that each node is either the root, an internal node or a leaf.
The fact that the tree is rooted is only for the convenience of the description.
In the construction, we will begin by assigning each node an initial size (see the definition of \textit{size} in Section \ref{sec:prelims}) which is possibly larger than $1$ (but at most $\poly(n)$). We will later show how to remove these variable sizes, and reduce to the case where all nodes have initial size $1$ (as in the original definition of HAC).

The basic idea of our construction is as follows. Our HAC instance will consist of a rooted tree where each child of the root corresponds to a column of $A$ in our \am instance. HAC merges will then happen in phases where each phase corresponds to a row of $A$. In a given phase, exactly one internal node will merge with the root which will correspond to this internal node's column being minimum for the corresponding row of $A$. In order to guarantee this, each child of the root will have its own carefully selected children such that merging with these children guarantees the desired behavior in every phase. 

More formally, the tree is constructed as follows.
The root $r$ of the tree has initial size $n^8$.
It has $n$ children, each of initial size $n^4$---we denote them by $v_0, \ldots, v_{n-1}$.
The root is connected to its children using edges of weight $1$, i.e., $w(r,v_i) = 1$ for all $i=0,1,\dots,n-1$ (thus, the normalized weight of the edges $\{rv_i\}_{i=0}^{n-1}$ are each $\frac{1}{n^{12}}$ at the start).
Each internal node has $n(n+1)$ children (leaf nodes) grouped into $n$ groups of $n+1$ leaves each.  We write $C_{i,j} = \{v_{i,j,0}, v_{i,j,1}, \dots, v_{i,j,n}\}$ to denote the $j$-th group of children of the $i$-th internal vertex. 
All the leaves $v_{i,j,k}$ have initial size $1$. Thus, the vertices in the full graph in our construction consists of the root $r$, internal nodes $\{v_0,v_1,\dots,v_{n-1}\}$ and leaves $\cup_{i=0}^{n-1} \cup_{j=0}^{n-1} \cup_{k=0}^n \{v_{i,j,k}\}$

Let $r_i = n^8 + i \cdot n^4$ (for $0 \leq i < n$).
For each pair of an internal node and each of its groups there are only two distinct edge weights for edges between the internal node and the leaves in the group.
Specifically, for an internal node $v_j$ and group $C_{j,i}$ 
of its children we have $A[i, j]+1$ edges of weight $\frac{1}{r_i - 1}$ and $n - A[i, j]$ edges of weight $\frac{1}{r_i + i \cdot n^3}$. Specifically, we set

\[ w(v_j, v_{j,i,k}) =\begin{cases} \frac{1}{r_i - 1} & \text{if } 0 \leq k \leq A[i,j] \\
\frac{1}{r_i + i \cdot n^3} & \text{if }  A[i,j] < k \leq n \\
\end{cases}\]

We call the two weights \emph{high-weight} and \emph{low-weight} edges, respectively. 
Note that their normalized weights are $\frac{1}{n^4(r_i - 1)}$ and $\frac{1}{n^4(r_i + i\cdot n^3)}$ respectively. Observe that the setting of \emph{high-weight} and \emph{low-weight} edges is independent of the internal node $v_j$, although the number of high versus low-weight edges depends on $A[i,j]$.
Moreover, note that even low-weight edges of any group $C_{j,i}$ have higher weights than high-weight edges of group $C_{j,i+1}$:
\begin{align*}
\frac{1}{n^4(r_i + i\cdot n^3)} & > \frac{1}{n^4(r_{i+1} + (i+1)\cdot n^3)} \iff \\
\frac{1}{n^8 + i \cdot n^4 + i\cdot n^3} & > \frac{1}{n^8 + (i+1)n^4 -1} \iff \\
i \cdot n^3 & < n^4 - 1.
\end{align*}
which follows from the fact that $i \leq n-1$.
We now demonstrate that the average linkage HAC on this instance works in $n$ phases numbered from $0$ to $n-1$, where in each phase $i$,
\begin{itemize}
\item $n-1$ internal nodes contract all of their incident group $i$ edges, and
\item one internal node contracts all of its high-weight edges to group $i$, after which it merges with the root.
\end{itemize}

Because of the internal node merging with the root, the root has incident leaves, but they are connected with edges of (normalized) edge weights $\leq \frac{1}{n^{15}}$ and so they will be irrelevant until all phases have been completed.
We  show that if we denote by $k_i$ the index of the internal node which merges with the root in phase $i$, the sequence $k_0, \ldots, k_{n-1}$ is a correct solution to the \am problem.


In order to analyze the algorithm, we prove the following claim.
For convenience, let us define $w_i = n^4 + i \cdot (n+1)$. 

\begin{claim}
In the beginning of phase $i$, the graph is as follows:
\begin{enumerate}
    \item The size of the root node is $r_i+ i^2 \cdot \Delta_i$ for some $\Delta_i \in [0, n+1]$.\label{l:rootweight}
    \item Exactly $i$ internal nodes have been merged with the root, and the corresponding values $k_0, \ldots, k_{i-1}$ have been computed correctly.\label{l:correct}
    \item The size of each of the $n-i$ remaining internal nodes is $w_i$.\label{l:intweight}
    \item For all remaining internal nodes, \emph{all} leaves in groups $0, \ldots, i-1$ have been merged into their parents, and no leaves in groups $i, \ldots, n-1$ have been merged.\label{l:groups}
    \item The root may have incident leaves (resulting from internal nodes contracting into it) connected to the root with edges of normalized weights $\leq \frac{1}{n^{15}}$.\label{l:leaves}
\end{enumerate}
\end{claim}
\begin{proof}
We prove the above claim using induction on $i$.
The base case of $i=0$ follows directly from how the tree is constructed.

We now simulate a single phase.
The edges between the root and the internal nodes have (normalized) weights $\frac{1}{w_i(r_i + i^2 \cdot \Delta_i)} > \frac{1}{n^{15}}$.
Hence, the additional leaf nodes incident to the root (see \cref{l:leaves} of the Claim) are irrelevant.
Thus, the highest weight edge in the graph is surely incident to one of the internal nodes.
Observe that the relative order of edge weights between an internal node $v$ and its children does not change as the leaves are merged into $v$.
Therefore, given that groups $0, \ldots, i-1$ do not exist anymore, among edges between the internal nodes and leaves, the edges of group $i$ have the highest weights.
In the beginning of a phase the high-weight edges in that group have normalized weights $\frac{1}{w_i(r_i -1)}$ and the low-weight edges have weight $\frac{1}{w_i (r_i + i \cdot n^2)}$.

Hence, we have that if we sort the edges by their normalized weights, the top 3 classes of edges are, starting from the highest weight:
\begin{enumerate}
    \item High-weight edges between internal nodes and leaves of group $i$.
    \item Edges between the root and the internal nodes.
    \item Low-weight edges between internal nodes and leaves of group $i$.
\end{enumerate}

\noindent We will show that the phase consists of the following sub-phases.
\begin{enumerate}
\item First, there is some number of subphases, where each of $n-i$ internal nodes contract one incident high-weight edge.
\item Then, there is exactly one subphase, where $n-i-1$ nodes contract an incident high-weight edge and one internal node merges with the root.
\item Then, the remaining $n-i-1$ internal nodes merge with all of their group $i$ leaves (we do not analyze the order in this subphase, as it is irrelevant).
\end{enumerate}

Assume that each internal node has at least one high-weight edge in group $i$.
Then, the algorithm will execute a Type 1 subphase: the first $n-i$ steps of the algorithm would merge exactly one high-weight edge incident to each internal node.
Note that when an edge incident to an internal node $v$ merges, the weight of $v$ increases, and so the incident edge weights decrease.
This guarantees that in the considered $n-i$ steps exactly one merge per internal node happens.

Type 1 subphases of $n-i$ steps continue as long as each each internal node has at least one high-weight group $i$ edge in the beginning of the subphase.
Each Type 1 subphase also causes the weight of each internal node to increase by $1$.
Clearly, since nodes are being merged into internal nodes, the ordering of edge weights incident to any internal node does not change.

At some point, in the beginning of a subphase there is an internal node that does not have any incident high-weight edge in group $i$.
Assume that this happened after $p$ Type 1 subphases have completed.
Thus, the size of each internal node is $w_i + p$.
Since by the construction each internal node had a different number of high-weight edges in group $i$, there is exactly one node $v$ with no high-weight incident edges and that node merges with the root.
This is when Type 2 subphase happens.
First, $n-i-1$ internal nodes contract with a high-weight incident group $i$ edge.
At this point the edge weights are as follows.
The weight of an edge between $v$ and the root is $\frac{1}{(w_i+p)(r_i + i^2\cdot \Delta_i)}$ and the weight of a high-weight edge in group $i$ is
$\frac{1}{(w_i+p+1)(r_i-1)}$. We have that the former is larger since
\begin{align*}
(w_i+p)(r_i + i^2\cdot \Delta_i) & < (w_i+p+1)(r_i-1) \iff\\
(w_i+p)\cdot i^2\cdot \Delta_i & < r_i-1 \iff\\
(n^4 + i \cdot (n+1) + p) \cdot i^2 \cdot \Delta_i & < n^8 + i \cdot n^4 - 1 \Leftarrow \\
(n^4 + n \cdot (n+1) + n) \cdot n \cdot n^2 & < n^8 + i \cdot n^3 - 1.
\end{align*}
Thus, the internal node with no incident high-weight edges in group $i$ merges with the root.
Observe that this is exactly the internal node which had the lowest number of high-weight edges in group $i$ among all remaining internal nodes.
This immediately implies that $k_i$ is computed correctly, proving \cref{l:correct}.

We now show that in the remaining part of the phase the $n-i-1$ remaining internal nodes contract their incident group $i$ edges.
First, observe that the new size of the root node is
\[
r_i + w_i + p + i^2 \cdot \Delta_i = \left(n^8 + i \cdot n^4 + n^4\right) + \left(i \cdot (n+1) + i^2 \cdot \Delta_i + p\right) = r_{i+1} + (i+1)^2 \Delta_{i+1}
\]
for some $\Delta_{i+1} \in [0, n+1]$.
Note that we use the fact that both $\Delta_i$ and $p$ are upper bounded by $n+1$, which implies $i \cdot (n+1) + i^2 \cdot \Delta_i + p \leq (i+1)^2(n+1)$.
This proves \cref{l:rootweight}.
Thus for an internal node of weight $w$, the weight of its edge to the root is  
\[
\frac{1}{w \cdot (r_{i+1} + (i+1)^2 \Delta_{i+1})} \leq \frac{1}{w \cdot r_{i+1}} = \frac{1}{w \cdot (n^8 + (i+1) \cdot n^4)}.\]
On the other hand, its low-weight edges to group $i$ leaves have weight
\[
\frac{1}{w(n^8 + i(n^4 + n^2))}.
\]
As a result, in the remaining part of the current phase all internal nodes will contract all their incident group $i$ edges. This implies \cref{l:groups}.
Thus, the size of each internal node within the phase  increases to $w_i + (n+1) = n^4 + i \cdot (n+1) + n + 1 = n^4 + (i+1)(n+1)$, as required.
This proves \cref{l:intweight} and completes the proof.
\end{proof}

Finally, we now claim that, given an instance of \am with input index $x \in [0,n)$ and  an algorithm which can compute the solutions to Problem \ref{prob:hac}, we can compute the solution $k_x$ to Problem \ref{l:am} in logspace. To see this, note that it suffices to determine the value $k_x$ as defined above given an algorithm for \hac. To see this, note that for any internal node $i$, we can query the \hac algorithm to determine which time step $t_i$ it merged with the root. 
This does not directly tell us which phase $i$ merged with the root, but for a given $i$ we can determine if it merged in phase $x$ 
by comparing $t_i$ with $t_j$ for all $j \in \{0,1,\dots,n-1\} \setminus \{i\}$, and checking if there are exactly $x-1$ values of $t_j$ smaller than $t_i$. This clearly can be verified in log-space. 
Repeating for all $i \in  \{0,1,\dots,n-1\}$ allows us to correctly determine the identity of the internal node that merged with the root in phase $x$, and therefore the value of $k_x$, in logspace as required.

To complete the proof of the Lemma it remains to show how to drop the assumption on the node sizes being initially not all equal to $1$.
In order to obtain a node of size $w$ it suffices to create a node of weight $1$ and initially connect it to $w-1$ auxiliary nodes using very high weight edges.
This will force the algorithm to merge all these auxiliary nodes and increase the size of that node to $w$.
Since the auxiliary leaves are connected only to the root and internal nodes (the leaves in our construction have weight $1$), the diameter of the tree does not increase.
\end{proof}

\section{Average Linkage HAC on Paths in \NC}\label{sec:pathupperbound}
In this section, we present an $\tilde{O}(n)$ work and $O(\mathsf{polylog}(n))$ depth algorithm for solving average linkage HAC on path graphs, provided that the aspect ratio of the input instance is bounded by $\poly(n)$.
The {\em aspect ratio} is defined as $\mathcal{A} = W_{\max}/W_{\min}$, where $W_{\max} = \argmax_{e \in E} w(e)$ and $W_{\min} = \argmin_{e\in E} w(e)$ (note that this definition excludes all non-edges, which implicitly have a weight of 0).
The main algorithm for this section is presented in Algorithm~\ref{alg:paths}, which we first give a high-level overview of next, followed by a detailed description.

\begin{algorithm}[ht]
\caption{Average linkage HAC on paths; $\mathtt{ProcessChain}$ assumes that the input chain has a certain structure. However, this assumption can be removed easily (see \Cref{sec:pathupperbound})}\label{alg:paths}
\SetKwInput{Input}{Input}
\SetKwFunction{FHac}{PathHAC}
\SetKwFunction{Fchain}{ProcessChain}
\SetKwProg{Fn}{Function}{:}{}
\SetKwProg{parfor}{parallel for}{:}{}

\Input{A path graph $G = (V, E, w)$}
\Fn{\FHac{$x, y, w, S$}}{
Compute buckets $B_1,B_2,\ldots,B_T$ such that $B_i$ contains edges with weights in $((2/3)^{t}w_{\max}, (2/3)^{t-1}w_{\max}]$.\\
\For{$t=1$ to $T$}{
Let $\mathcal{C} \gets$ Nearest-neighbor chains on the induced graph $G[B_t]$.\\
\parfor{each $C \in \mathcal{C}$}{
\Fchain{$C$}
}
}
}
\BlankLine
\Input{A chain $C=(c_1,c_2,\ldots,c_N)$ such that $c_{i-1}$ is the nearest neighbor of $c_i$ and $(c_1,c_2)$ is a reciprocal pair.}
\Fn{\Fchain{$C,t$}}{
Split the chain $C$ at all indices $i$ such that $S(c_i) < 2S(c_{i-1})$, except $i=2$.\\
We get a set of contiguous \emph{subchains} $C_1,C_2,\ldots,C_K$. Subchain $C_j$ falls into one of two categories: 
(A) the nearest clusters in $C_{j-1}$ and $C_j$ merge, or 
(B) no two clusters from $C_{j-1}$ and $C_j$ merge.\\
Find the merge categories for all subchains to obtain the modified subchains $C_1',C_2',\ldots, C_K'$.\\
\parfor{each $j=1$ to $k$}{
Run average linkage HAC on $C_j'$ until no edge remains with weight within the threshold of $B_t$.
}

}
\end{algorithm}

In average linkage HAC, the weight (i.e., similarity) of edges monotonically decreases over time. 
Thus, our idea is to partition the edges into \emph{buckets} where the edges in any bucket have the same similarity, up to constant factors. 
Next, we process these buckets in {\em phases}, from the highest similarity bucket to the lowest. 
In each phase, we perform a modified version of the classic nearest-neighbor chain algorithm (\Cref{alg:nnchain}), wherein we compute the nearest-neighbor chains for the graph induced on the edges in that bucket, and process each chain independently. 
We note that when we use the terminology {\em nearest neighbor} of a vertex in what follows, we refer to the neighbor along the {\em highest weight edge} incident to the vertex.

Initially, each cluster is a singleton, and we might end up with $\Omega(n)$ sequential dependencies to resolve. 
However, we observe that in this special case when the size of every cluster is equal, starting with the reciprocal pair, every {\em alternate} edge in this chain can be merged independently, and the rest of the edges will be moved to a later bucket. 
We can compute the edges that will be merged easily via a simple $\mathsf{prefix\text{-}sum}$ routine~\cite{blellochscan}. 
However, when the cluster sizes are arbitrary, this observation no longer holds.
Nonetheless, we show that we can partition each chain further into $O(\log n)$-sized subchains such that, even though the dependencies within a subchain must be resolved sequentially, the dependencies across subchains can be resolved in parallel using a similar application of $\mathsf{prefix\text{-}sum}$. 
In this section we show (1) that our parallel algorithm is highly efficient (it runs near-linear time in the number of nodes) and runs in poly-logarithmic depth and (2) that our algorithm implies that the dendrogram height of a path input with polynomial aspect ratio is always poly-logarithmic.


\paragraph{Implementation Details and Correctness}
We now elaborate on certain implementation aspects of \Cref{alg:paths}. In \Cref{alg:paths}, there are two implicit assumptions made about the chains:
\begin{itemize}
    \item Every edge in $B_t$ is present in some chain in $\mathcal{C}$.
    \item Each chain can be represented as $(c_1,c_2,\ldots, c_N)$ such that $c_{i-1}$ is the nearest neighbor of cluster $c_i$, and $(c_1,c_2)$ is the \emph{reciprocal pair}. 
    Thus, the weights are in non-increasing order from left to right.
\end{itemize}
We will see later how to remove these assumptions. By the above assumptions, each chain will be a separate connected component in the graph induced on $B_i$, thus, allowing us to process each chain independently. In $\mathtt{ProcessChain}$, we split each chain at all indices such that $S(c_i) < 2 S(c_{i-1})$ giving us a set of subchains. This implies that for two adjacent clusters $c_{j-1},c_{j}$ in the same subchain, $S(c_j) \ge 2S_{j-1}$ (see \Cref{fig:path-subchains}). Thus, the size of a subchain cannot exceed $\log n$, and we can afford to process these subchains sequentially, while still obtaining low depth overall. 

\begin{figure}[ht]
    \centering
    \includegraphics[width=0.9\textwidth]{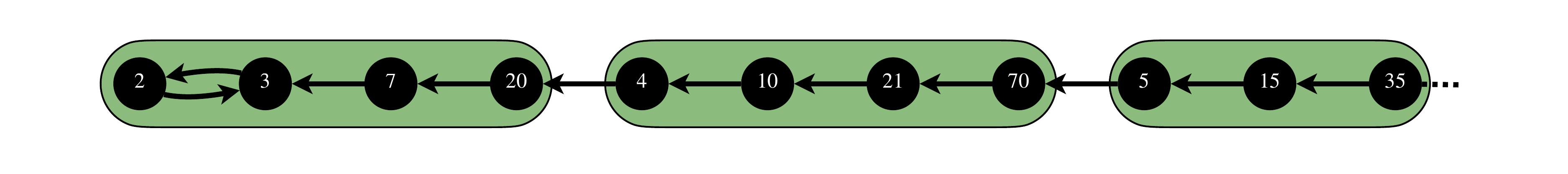}
    \caption{An example describing the partition of a chain into subchains based on cluster sizes.}
    \label{fig:path-subchains}
\end{figure}

We now prove certain key properties about subchains. Let $u_0,u_1,u_2$ denote three contiguous clusters in a chain such that $u_0$ is the last cluster in its subchain, and $u_1$ is the first cluster of its subchain. $u_2$ may or may not belong to the same subchain as $u_1$.

\begin{lemma}\label{lem:path_A}
    If cluster $u_1$ first merges with the cluster (containing) $u_0$, then the weight of the edge $(u_1,u_2)$ reduces to at most $2/3$ times its previous weight.
\end{lemma}
\begin{proof}
    Let $x = w(e)/(S(u_1)S(u_2))$ denote the weight of edge $e=(u_1,u_2)$. If cluster $u_1$ merges with a cluster in its preceding subchain, the merged cluster will be of size at least $S(u_0)+S(u_1)$. Let $x'$ denote the new weight of edge $e$ after this merge. Then,
    \begin{align*}
        x' \le \frac{w(e)}{(S(u_0)+S(u_1))S(u_2)} \implies \frac{1}{x'} &\ge \frac{1}{x} + \frac{S(u_0)S(u_2)}{w(e)}\\
        \implies \frac{1}{x'} &\ge \frac{1}{x} + \frac{S(u_1)S(u_2)}{2w(e)}, & \text{since $S(u_0) \ge S(u_1)/2$}\\
        \implies \frac{1}{x'} &\ge \frac{3}{2x},
    \end{align*}
    Therefore, $x' \le 2x/3$.
\end{proof}

\begin{lemma}\label{lem:path_B}
    If cluster $u_1$ first merges with the cluster (containing) $u_2$, then the weight of the edge $(u_0,u_1)$ reduces to at most $1/3$ times its previous weight.
\end{lemma}
The proof follows by a similar argument as in \Cref{lem:path_A}. By \Cref{lem:path_A,lem:path_B}, each subchain falls into one of the following categories:
\begin{enumerate}[label=(\Alph*)]
    \item The first cluster of the subchain merges with a cluster in the preceding subchain (if exists).
    \item The first cluster of the subchain doesn't merge with a cluster from the preceding subchain.
\end{enumerate}
In either case, one of the edges incident on the first cluster of the subchain moves to a later bucket after it merges, disconnecting the two subchains in the induced graph. Thus, if a subchain falls into category (A), we can move its first cluster to the preceding subchain; if it falls into category (B), we make no changes. By this process, we obtained the modified subchains $C_1',C_2',\ldots,C_K'$, and by the above argument we can run average linkage HAC independently on them. Note that we do not run HAC to completion at these subchains, rather until we have edges with weight within the threshold defined by its bucket. 

Note that the first subchain of a chain will always fall into category (B) since it contains a reciprocal pair which is guaranteed to merge.
Running HAC on this subchain (starting with merging the reciprocal pair) determines whether the next subchain falls into the merge category (A) or (B), which in turn determines the merge category for the following subchain, and so on. 
We now show that we can propagate this merge behavior across subchains with the help of a simple prefix-sum routine in $\poly(\log n)$ depth.

\begin{lemma}\label{lem:path_prefix_sum}
    Given a chain of size $N$ partitioned into subchains $C_1,C_2,\ldots,C_K$, we can determine the merge category ((A) or (B)) for each subchain in $\tilde{O}(N)$ work and $\poly(\log n)$ depth.
\end{lemma}
\begin{proof}
    For subchain $C_j$, let $f(j,A)$ denote the resulting merge category of $C_{j+1}$ if $C_j$ belongs to category (A). Define $f(j,B)$ similarly. Firstly, recall that the category of $C_1$ is determined (it is always category (B)). 
    If $f(j,A) = f(j,B)$, for some $j$, this implies that the category of $C_{j+1}$ is determined regardless of the category of $C_j$. Thus, we can split the chain at subchain $C_{j+1}$ and deal with the two chains obtained, independently. We now assume that $f(j,A)\ne f(j,B)$ for all $j$.

    We create an array $A$ of size $K$ such that,
    \begin{align*}
        A[j] = \begin{cases}
            0 & \text{if $j=1$ and $f(1,B)=A$},\\
            1 & \text{if $j=1$ and $f(1,B)=B$},\\
            0 & \text{if $j>1$ and $f(j,A)=A$},\\
            1 & \text{if $j>1$ and $f(j,A)=B$}.
        \end{cases} 
    \end{align*}
    Let $B= \mathsf{prefix\text{-}sum}(A,\mathtt{XOR},0)$, i.e., $B[1]=0$ and $B[i]=A_1\oplus A_2 \oplus\ldots A_{i-1}$.
    In other words, we compute the prefix sum of the array $A$ with respect to the $\mathtt{XOR}$ function, using $0$ as the left identity element.
    Then, we claim that the merge category of subchain $C_j$ is (A) if $B[j]=0$, otherwise it is (B).

    To see this, consider subchains such that $f(j,A)=A$. These subchains propagate the same category (as them) to the next subchain, which can be viewed as an $\mathtt{XOR}$ operation with a $0$. Similarly, subchains that have $f(j,A)=B$ propagate the opposite category (as them) to the next subchain, which can be viewed as an $\mathtt{XOR}$ with $1$. $A[1]$ denotes the category of subchain $C_2$ that is determined by $C_1$. Thus, computing the $\mathsf{prefix\text{-}sum}$ correctly computes the categories for each subchain.

    The values of $f(j,A)$ (and $f(j,B)$) can be computed by simply simulating average linkage HAC sequentially at each subchain, assuming the category (A) (and (B)) at that subchain. 
    Since the subchain sizes are at most $\log n$, the work incurred to solve HAC on a subchain of size $r$ will be $O(r\log\log n)$.
    
    To see this, consider running Algorithm~\ref{alg:heapbased} on the chain, which we claim runs in $O(N\log\log n)$ work (and depth).
    This is because the heap-based algorithm for HAC (described in more detail in Section~\ref{sec:workupperbound}) on a path containing $r \leq \log n$ elements operates over a heap containing at most $\log n$ elements, and thus each of the at most $r$ merge operations require $O(\log \log n)$ time and only require updating the two neighboring edges incident to the merged edge.
    Since the algorithm used on each subchain is sequential, the depth at each subchain will also be $O(r\log\log n) \in O(\log n\log\log n)$. 
    The $\mathsf{prefix\text{-}sum}$ operation runs in $O(K)$ work and $O(\log K)$ depth. 
    
    Thus, the overall work for a chain containing $N$ elements will be $O(N\log\log n)$ and depth will be $O(\log n\log\log n)$.
\end{proof}

Since the edges are processed in non-increasing order of weight (due to bucketing), and \Cref{lem:path_A,lem:path_B,lem:path_prefix_sum} correctly computes the subchains and proves the correctness of processing these subchains independently, the overall correctness of the algorithm follows directly.

\paragraph{Work and Depth Bounds}
We will now argue the work and depth bounds of the overall algorithm.
Observe that, since the minimum possible weight is at least $W_{\min}/n^2$, the total number of phases will be $T \in O(\log \mathcal{A} + \log n) \in O(\log n)$. 
In each phase, we can compute the nearest neighbor chains in $O(n)$ work and $O(\log n)$ depth; the subchains can also be computed with similar work-depth bounds. By \Cref{lem:path_prefix_sum}, the total work incurred in finding the merge categories for each subchain and running average linkage HAC on these subchains, across all chains, in the worst case will be $O(n\log\log n)$ work and $O(\log n\log\log n)$ depth. Thus, the overall work of the algorithm is $O(n\log n\log\log n)$ and depth is $O(\log^2n\log\log n)$. 

Thus, we have the following theorem about the work and depth of the algorithm:
\mainUpperPar*

\paragraph{Resolving assumptions about chain structure}
\begin{figure}[ht]
    \centering
    \includegraphics[width=\textwidth]{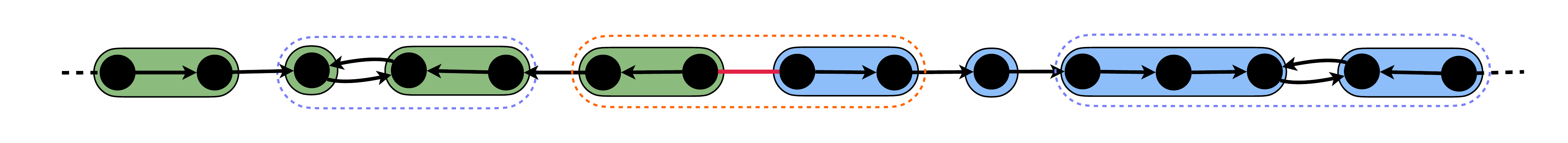}
    \caption{An example illustrating the general (possible) structure of the nearest neighbor chains computed in each phase of \Cref{alg:paths}. The red edge here is an example of an edge that is not present in any chain. The dotted lines represents the various subchain merges to address the assumptions made about the nearest neighbor chain structure in \Cref{alg:paths}.}
    \label{fig:path-general-chains}
\end{figure}
We will now address the assumptions made about the structure of chains in the algorithm. When the nearest neighbor chains are computed, the structure of the chains can be represented as 
$$(c'_{N'}\ldots, c_3',c_2', c_1',c_1,c_2,c_3,\ldots,c_{N}),$$
where $(c_1',c_1)$ is the reciprocal pair, $c'_{i-1}$ is the nearest neighbor of $c'_i$, and $c_{i-1}$ is the nearest neighbor of $c_i$. We can extend the idea of partitioning these chains into subchains as follows: partition the chains $(c_1,c_2, \ldots, c_N)$ and $(c'_1,c'_2,\ldots,c'_{N'})$, independently, into subchains using the same criterion as before. Now, merge the subchains containing $c_1$ and $c'_1$ resulting in a subchain of length at most $2\log n$ (see \Cref{fig:path-general-chains}). Then, independently propagate the merge behavior across these two sets of subchains, with the subchain containing $(c_1,c'_1)$ as the first subchain. 

Finally, when we compute the nearest neighbor chains, it is possible for some edges in the bucket to not be present in any chain. However, these edges can occur only between the last vertices of two chains (see \Cref{fig:path-general-chains}). Thus, after finding the merge categories for the subchains containing these vertices (in their respective chains), we can merge these two subchains along with the edge (whose size will be at most $4\log n$ in the worst case). It is not hard to extend the correctness and running time arguments when these modifications are incorporated.

\paragraph{Height of the dendrogram is \texorpdfstring{\boldmath$O(\log^2n)$}{O(log^2(n))}}
An interesting consequence of our algorithm is a constructive proof that the height of the dendrogram obtained when we run average linkage HAC on paths with polynomial aspect ratio is $O(\log^2n)$.

\begin{restatable}{theorem}{pathheight}\label{thm:pathheight}
    Average linkage HAC on path graphs with $\poly(n)$ aspect-ratio returns a dendrogram with height at most $O(\log^2n)$.
\end{restatable}

\begin{proof}
    Starting with phase 1, the algorithm runs average linkage HAC (partially) on independent $O(\log n)$-sized contiguous portions (i.e., subchains) of the path. Thus, each subchain will correspond to independent parts in the output dendrogram. Since there are $O(\log n)$ phases and the max-height generated in each phase is at most $O(\log n)$, we get an overall bound of $O(\log^2n)$ for the height in the worst case.
\end{proof}

\begin{algorithm}[!t]\caption{\small The nearest-neighbor chain algorithm for HAC.}\label{alg:nnchain}
\KwIn{$G = (V, E, w)$}

\SetKwFunction{FNNChain}{NearestNeighborChain}
\SetKwProg{Fn}{Function}{:}{}

\Fn{\FNNChain{$G$}}{
$\clustering \gets$ clustering where each $v \in V$ is in a separate cluster and is {\em active}.\\
\For{each cluster $v \in V$}{
  \If {$v$ is active}{
    Initialize a stack $S$ initially containing only $v$ \\
    \While {$S$ is not empty}{
      $t \gets \textsc{Top}(S)$.\\
      $(t, b, w_{t,b}) \gets \textsc{BestEdge}(t)$.\label{line:nnchainbestedge} \\
      \If{$b$ is already on $S$}{
        $\textsc{Pop}(S)$. \\
        $\textsc{Merge}(t, \textsc{Top}(S))$.\label{line:nnchainmerge}  \ \  //\ {\em $t$ is marked as inactive; $\textsc{Top}(S)$ remains active.}\\
        $\textsc{Pop}(S)$.
       } \Else {
        $\textsc{Push}(S, b)$.
       }
    }
  }
}
}
\end{algorithm}

\section{\texorpdfstring{\boldmath$O(mh\log n)$}{O(mhlogn)} Upper Bound}\label{sec:workupperbound}


In this section we show that we can solve graph-based average linkage HAC in $\tilde{O}(mh)$ time where $h$ is the {\em height} of the output dendrogram.
Thus, when $mh = o(n^{3/2})$ this algorithm improves over the lower bound from Section~\ref{sec:worklb}.
Interestingly, we give a simple analysis showing that either of the two classic approaches to HAC---the nearest-neighbor chain or heap-based algorithms---achieves this bound.
We give the pseudocode for the nearest-neighbor chain algorithm in Algorithm~\ref{alg:nnchain} and the heap-based algorithm in Algorithm~\ref{alg:heapbased}.
In a nutshell, in both algorithms initially all vertices start in their own cluster, and are {\emph{active}}.
As the algorithms proceed, they find edges to merge that the exact greedy HAC algorithm (Algorithm~\ref{alg:staticgraph}) will make and merge them, marking one of the endpoints as {\em inactive}.
This process continues until no further edges between active clusters remain.
We discuss the algorithms individually in more detail in what follows.

We start by discussing some data structures and details in how the algorithms carry out merges.

\paragraph{Neighborhood representation}
The neighbors of each vertex are stored in a max-heap supporting standard operations (e.g., \textsc{find-min}, \textsc{delete-min}, \textsc{insert}, \textsc{decrease-key}).
The priority of a neighbor is the average linkage weight of the edge to the neighbor.
At the start of the algorithm, the heaps are initialized with the initial neighbors of each vertex.
Let $N(A)$ represent the neighbors of a cluster $A$.
For concreteness in what follows, we use Fibonacci heaps~\cite{CLRS, fredman1987fibonacci}.

\paragraph{Merging clusters}
Next, we discuss how to implement the {\em merge} step in both algorithms, which merges two clusters $A$ and $B$, and their corresponding heaps (i.e., merging $N(A)$ and $N(B)$ and updating the average linkage weights of {\em all} edges affected by the merge).
We maintain the invariant that the heaps always store the correct average linkage weights of all edges.
We first compute $\mathcal{I}(A,B) = N(A) \cap N(B)$, which can be done by sorting both sets and merging the sorted arrays.
For each $C \in \mathcal{I}(A, B)$, without loss of generality, we remove the reference to $A$ in the heap and update the priority of $B$ to reflect the new average linkage weight from $C$ to $A \cup B$ using \textsc{decrease-key}.
Similarly, we update each $C \in N(A) \setminus \mathcal{I}(A, B)$ to reference $B$ instead of $A$.
Finally, for each $C \in N(A) \cup N(B)$ we update the priority of the edge $(C, B)$ to reflect the new size of $B$.

At the end of the merge: (1) $N(B)$ contains $N(A) \cup N(B) \setminus \{A, B\}$; (2) all edges in the updated $N(B)$ have their average linkage weight correctly set; and (3) after the update all $C \in N(B)$ no longer reference $A$, but point to $B$ and have their average linkage weights correctly set.
%
The cost of this merge procedure is dominated by the cost of sorting $N(A)$ and $N(B)$; plugging in the cost bounds for any reasonable heap data structure (e.g., a Fibonacci heap) yields the following lemma:
\begin{lemma}\label{lem:merge}
Merging two clusters $A,B$ with $N(A) + N(B) = T$ can be done in $O(T \log n)$ time.
\end{lemma}

\paragraph{Nearest-Neighbor Chain Analysis}
As we run the nearest-neighbor chain algorithm (Algorithm~\ref{alg:nnchain}), the algorithm only requires two non-trivial operations beyond basic data structures (e.g., stacks and heaps):\footnote{In fact the same two primitives are also the only non-trivial ones used for the heap-based algorithm (Algorithm~\ref{alg:heapbased}) and thus we fully specify the details for this algorithm as well.}
\begin{enumerate}
  \item \textsc{BestEdge}($C$): fetch the highest similarity edge incident to a cluster $C$.
  \item \textsc{Merge}($A, B$): merge the clusters $A,B$ into a cluster $A \cup B$ (represented, without loss of generality by $B$). Thus after the merge, $A$ is inactive, and $B$ remains active.
\end{enumerate}

Using the data structures above, we can implement both operations very quickly. 
In particular, Lemma~\ref{lem:merge} gives the time for merging two clusters, and \textsc{BestEdge}($C$) can be implemented in $O(1)$ time using the \textsc{find-min} operation.

After merging two clusters that are reciprocal ``best'' neighbors, the algorithm back-tracks to the third vertex in the chain and restarts the process, using calls to \textsc{BestEdge}.
The cost of these calls can be charged to the merges, resulting in a total of $O(n)$ time for the $n-1$ merges.
The time that still must be accounted for is that of \textsc{Merge}.

\begin{claim}\label{claim:mergecost}
The total amount time spent to perform all merges made by the nearest-neighbor chain algorithm is $O(mh\log n)$.
\end{claim}
\begin{proof}
Split each edge $(u,v)$ into two directed edges $(u,v)$ and $(v,u)$, since the edge will be present initially in both $N(u)$ and $N(v)$ (and the intermediary clusters they form until they are merged).

Consider a directed edge $(u,v)$, and the merges it experiences as it flows along the path from the initial cluster containing it ($u$) to the cluster where it is finally merged.
The dendrogram that the algorithm constructs has height $h$ by definition, and thus the length of this path is at most $h$.
In each node on the path, the edge contributes a cost of $O(\log n)$ to the merge cost in Lemma~\ref{lem:merge}.
Thus, the cost of this edge across the entire path is at most $O(h \log n)$.

Applying this charging argument for all $2m$ directed edges, the total cost over all merges is $O(mh\log n)$.
\end{proof}

Putting together the above yields the following theorem.
\mainUpperSec*

\begin{algorithm}[!t]\caption{The heap-based algorithm for HAC.}\label{alg:heapbased}
\KwIn{$G = (V, E, w)$}

\SetKwFunction{FHeapBased}{HeapBasedHAC}
\SetKwProg{Fn}{Function}{:}{}

\Fn{\FHeapBased{$G$}}{
$\clustering \gets$ clustering where each $v \in V$ is in a separate cluster and is {\em active}.\\
Let $H$ be a max-heap storing the highest-weight edge incident to each active cluster in the graph \\
\While{$|H| > 1$}{
  Let $e=(u, v, w_{u,v})$ be the max weight edge in $H$. \\
  Delete $e$ from $H$. \\
  \If {$v$ is inactive} { \label{line:heapbadextract}
    Let $e' = (u, v', w_{u,v'}) = \bestedge{u}$. \\
    Insert $e'$ into $H$. \label{line:heapinsertfirst} \\
  } \Else {
    $x = \textsc{Merge}(u, v)$.    \ \  //\ {\em $u$ is marked as inactive; $v$ remains active.} \\
    Let $e' = (x, y, w_{x,y}) = \bestedge{x}$. \\
    Insert $e'$ into $H$. \\
  }
}
}
\end{algorithm}

\paragraph{Heap-Based Algorithm}

We note it is easy to extend the argument above to obtain a similar bound for the heap-based algorithm 
(pseudocode shown in Algorithm~\ref{alg:heapbased}).
In particular, beyond the cost for merging, which is the same in the nearest-neighbor chain algorithm and the heap-based algorithm (captured by Claim~\ref{claim:mergecost}), the only extra work incurred by the heap-based algorithm is due to edges unsuccessfully extracted from the heap (Line~\ref{line:heapbadextract} of Algorithm~\ref{alg:heapbased}). 
The number of such unsuccessful edges is at most $m$, and each such edge incurs $O(\log n)$ work due to an extra call to \bestedge{u}, for a total of $O(m\log n)$ extra work.
Thus the total work of the heap-based algorithm is also $O(mh\log n)$.

\begin{theorem}
There is an implementation of the heap-based algorithm for average linkage HAC that runs in $O(mh \log n)$ time where $h$ is the height of the output dendrogram.
\end{theorem}


\section{Conclusion}
In this paper, we studied the parallel and sequential complexity of hierarchical
graph clustering. We gave new classic and fine-grained reductions for Hierarchical Agglomerative Clustering (HAC) under the average linkage
measure that likely rule out efficient algorithms for {\em exact} average linkage, parallel or otherwise. 
We also showed that such impossibility results can be circumvented if the output dendrogram has low height or is a path. An interesting question is whether such structure can be leveraged for other variants of interest of average linkage HAC: for example, can we can obtain dynamic algorithms for HAC that are also parameterized by the height?

\newpage
\bibliography{main}

%

\end{document}